\documentclass[article]{amsart}

\usepackage[round]{natbib}
\usepackage[a4paper, total={6in, 8in}]{geometry}

\usepackage[english]{babel}
\usepackage{blindtext}
\usepackage{a4wide}
\usepackage[T1]{fontenc}
\usepackage{lscape}
\usepackage{amsaddr}

\usepackage{graphicx}
\usepackage{amsfonts,amsmath,amssymb,mathtools}
\usepackage{enumitem}
\usepackage{booktabs}
\usepackage{comment}
\usepackage{xcolor}
\usepackage[linesnumbered,ruled,vlined]{algorithm2e}
\SetKwInput{KwSetup}{Setup}                
\SetKwInput{KwInput}{Input}                

\SetCommentSty{mycommfont}

\usepackage{hyperref}
\newcommand{\p}{\mathbb{P}}

\newcommand{\e}{\mathbb{E}}

\newcommand{\reals}{\mathbb{R}}
\newcommand{\ind}{\mathbf{1}}

\newcommand{\diff}{\,\mathrm{d}}
\newcommand{\ttau}{\tilde{\tau}}
\newcommand{\xtau}{\tau^x_{t}}
\newcommand{\xtaus}{\tau^x_{s}}

\newcommand{\Sr}{\mathcal{S}} 
\newcommand{\Cr}{\mathcal{C}} 

\newcommand{\abs}[1]{\left\lvert#1\right\rvert}
\renewcommand{\tilde}{\widetilde}
\DeclareMathOperator*{\esssup}{ess\,sup}

\usepackage{graphicx}
\newtheorem{theorem}{Theorem}[section]
\newtheorem{proposition}{Proposition}[section]
\newtheorem{lemma}{Lemma}[section]
\newtheorem{corollary}{Corollary}[section]
\newtheorem{assump}{Assumption}[section]
\newtheorem{remark}{Remark}[section]
\newtheorem{example}{Example}[section]
\newtheorem*{proposition*}{Proposition}
\theoremstyle{definition}
\newtheorem{defi}[theorem]{Definition}

\allowdisplaybreaks
\allowdisplaybreaks
	\title[On an Optimal Stopping Problem with a Discontinuous Reward]{On an Optimal Stopping Problem with a Discontinuous Reward}

\author{Anne MacKay$^{\ast\dag\ddag}$ and Marie-Claude Vachon$^\ddag$}
\address{$\dag$Universit\'e de Sherbooke, Sherbrooke, Qu\'ebec, Canada\\
$\ddag$ Universit\'e du Qu\'ebec \`a Montr\'eal, Montr\'eal,  Qu\'ebec, Canada\\
}
\thanks{$^\ast$Corresponding author. Email: Anne.MacKay@USherbrooke.ca\\}

\date{\today}

\begin{document}
	
\begin{abstract} 
We study an optimal stopping problem with an unbounded, time-dependent and discontinuous reward function. 
This problem is motivated by the pricing of a variable annuity contract with guaranteed minimum maturity benefit, under the assumption that the policyholder's surrender behaviour maximizes the risk-neutral value of the contract. 
We consider a general fee and surrender charge function, and give a condition under which optimal stopping always occurs at maturity.
Using an alternative representation for the value function of the optimization problem, we study its analytical properties and the resulting surrender (or exercise) region.
In particular, we show that the non-emptiness and the shape of the surrender region are fully characterized by the fee and the surrender charge functions, which provides a powerful tool to understand their interrelation and how it affects early surrenders and the optimal surrender boundary. 
Under certain conditions on these two functions, we develop three representations for the value function; two are analogous to their American option counterpart, and one is new to the actuarial and American option pricing literature.
	
	
	 \vspace{0.5cm}
	\vspace{1cm}
\end{abstract}
	\maketitle
	
	\newpage
\section{Introduction}
\label{intro}
Variable annuities (VA) are structured products sold by insurance companies that are mainly used for retirement planning. An initial premium is deposited in an investment account (or fund) whose return is linked to that of one or more risky assets. They are similar to mutual funds, but they also offer financial guarantees at the end of a pre-determined accumulation period. The embedded guarantees are akin to long-dated options, but they are funded by a periodic fee, generally set as a percentage of the investment account value, rather than being paid upfront. In this paper, we focus on a guaranteed minimum maturity benefit (GMMB), which provides the policyholder a minimum guaranteed amount at maturity of the contract.  Another common feature of variable annuity contracts is that policyholders generally have the right to surrender, or lapse, their contract prior to maturity.  When policyholders choose to do so, they receive the value accumulated in the investment account, reduced by a penalty for early surrender (or surrender charge).  
The uncertainty faced by VA providers with respect to early termination is known as surrender risk. Early surrenders can entail negative consequences such as liquidity issues. Thus, incorporating adequate surrender assumptions in the pricing of VAs is essential for risk management purposes; see \citet{niittuinpera2020IAA}.

Different early surrender modelling approaches have been proposed in the literature (see \cite{bauer2017policyholder} and \cite{feng2022variable} for a review), from the use of utility functions, \cite{gao2012optimal}, to modern statistical techniques, \cite{zhu2015statistical}. Early surrenders are also often expressed as a decision taken by policyholders on a strictly rational basis, meaning that the contract is terminated as soon as it is optimal to do so from a risk-neutral perspective (see \cite{Bernard2014}, \cite{jeon2018optimal}, \cite{kang2018optimal}, \cite{MackayVachonCui2022VaCTMC}, and \cite{milevsky2001real}, among others). \cite{bauer2017policyholder} discusses the impact of other factors on policyholder behavior, such as taxes and expenses. This article gave rise to another stream of literature in which market frictions such as taxation rules are considered (\cite{alonso2020taxation}, \cite{bauer2023cheaper}, and \cite{moenig2016revisiting}), which helps to explain the discrepancies between market and model fee rates. Other recent studies include lapse and reentry strategies in their analysis (\cite{bernard2019less} and \cite{moenig2018lapse}), whereas \cite{moenig2021economics} consider a third-party investor to whom the policyholder can sell her contract.  

The goal of this paper is to study the properties of the value function of the VA contract under the assumption that the policyholder maximizes its risk-neutral value, when the fee and the surrender penalty (or charge) are both time and state-dependent (that is, depending on the value of the VA account, see \cite{Bernard2014}).
Our setup includes, among others, the constant fee case, the state-dependent fees of \cite{Bernard2014}, \cite{delong2014pricing}, and \cite{Mackay2017}, and the time-dependent fees of \cite{bernard2019less} and \cite{kirkby2022valuation}. 
Under this assumption, valuing a variable annuity contract is equivalent to solving an optimal stopping problem similar to pricing an American option. However, the financial guarantee being applied only at maturity in a VA contract creates a discontinuity in the reward function at maturity, which sets it apart from the continuous reward function of the American option and complicates the optimal stopping problem involved in the pricing of VAs.  
The impact of the surrender charge and the fee structure on the optimal surrender strategy has been studied in the Black-Scholes framework by \cite{Bernard2014}, \cite{BernardMackay2015}, \cite{bernard2019less}, \cite{Mackay2014}, \cite{Mackay2017}, and \cite{moenig2018lapse}, and in a more general setting by \cite{kang2018optimal} and \cite{MackayVachonCui2022VaCTMC}. However, these articles approach the problem from a numerical perspective.
Recently, and independently from our work, \cite{de2024variable} study the optimal surrender problem in VAs from a theoretical perspective. 
While their setting includes mortality, which ours does not, we consider a more general structure for the fee and surrender charge functions. 
The analysis presented in this paper is therefore significantly different from theirs, but a parallel can be drawn between some of their results and ours.

Recently, \cite{luo2021optimal} studied variable annuity contracts under regime-switching volatility models from an optimal stopping perspective. However, they consider a different surrender benefit, which results in a continuous reward function similar to that of a put option. \cite{chiarolla2020analytical} perform an analysis similar to ours, but for participating policies with minimum rate guarantee and surrender option. Participating policies are akin to variable annuities in that the premium paid by the policyholder tracks a financial portfolio, subject to a minimum rate guarantee. However, the problem they consider is fundamentally different from ours in multiple ways. i) In variable annuity contracts, guarantees are funded via periodic fees set as a percentage of the account value. This creates a discrepancy between the fee amount and the value of the financial guarantee, which becomes an incentive for the policyholder to surrender the policy early (\cite{milevsky2001real}). In participating policies, there is no such ongoing fee; upon termination or at maturity, the policyholder receives the reserve and a given percentage of the surplus, defined as a percentage of the tracking portfolio over the policy reserve (the intrinsic value). ii) In variable annuities, the surrender value differs from that of the maturity payout, which creates a time discontinuity in the reward function at maturity; whereas in participating policies, the intrinsic value is paid at maturity or upon early termination, so that the reward function is independent of time and continuous. iii) In participating policies, the contract is terminated if the underlying portfolio falls below the reserve. There is no such feature in a variable annuity, that is, early termination is at the sole discretion of the policyholder. 
The resulting optimal stopping problem we study in this paper is thus very different from the one considered by \cite{chiarolla2020analytical}, even if they are motivated by similar insurance products.

The time discontinuity and the unboundedness of the reward function involved in the pricing of variable annuities with GMMB prevent the simple application of results from the American option pricing literature to our problem.
For example, to express the value function as the solution to a free boundary value problem, one needs to establish its continuity. 
However, continuity of the value function often follows from the continuity of the reward function, which, in our setting, is unbounded and only upper semi-continuous, making the results usually cited in the context of American options inapplicable (see, for instance,  \cite{bassan2002optimal}, \cite{bensoussan2011applications}, \cite{de2019lipschitz}, \cite{jaillet1990variational}, \cite{krylov2008controlled}, and \cite{Lamberton2009}).   
In this paper, following the work of \cite{van1974optimal} and \cite{palczewski2010finite} (in the context of impulse control problems), we establish continuity of the value function by showing that it admits an alternative representation in terms of a continuous reward function. 
We can then confirm that it is a solution to a free boundary value problem and prove further properties, 
which provides enough regularity to apply a generalized version of Itô's formula to the value process. 
This allows us to derive various integral representations for the value function, thus
extending the results of \cite{bernard2014optimal} to general time-dependent fee and surrender charge functions and introducing a new decomposition, the continuation premium representation, to the literature on VA pricing.

The shape and the (non-)emptiness of the exercise region of an American option with a time-homogeneous reward function has first been studied rigorously by \cite{villeneuve1999exercise}. 
Other authors, such as \cite{jonsson2006threshold}, \cite{kotlow1973free}, and \cite{villeneuve2007threshold}, also studied the conditions under which the exercise region has a specific shape. 
These results are however not directly applicable to the problem we study in this paper, because of the time-dependence and the discontinuity of our reward function.
Under the assumption of a general fee and surrender charge structure, we are nonetheless able to identify a condition, expressed as a partial differential inequality, under which the optimal stopping strategy is always to hold on to the contract until maturity, thus mitigating surrender risk.
This condition provides a powerful tool to understand the interplay between the fees and surrender charges, and its effect on the optimal stopping strategy. 
It can also be used to further study and characterize the shape of the optimal exercise, or surrender, region.
In particular, we show that the surrender region can be a disconnected set and that the optimal surrender boundary can be discontinuous,
thus further highlighting the complexity of this optimal stopping problem.
We also extend the work of \cite{milevsky2001real}, who consider an infinite horizon, to a more general finite horizon problem.

Our results on the shape of the surrender region also shed light on the link between our original optimal stopping problem and its continuous reward function counterpart.
In particular, we obtain conditions under which the two reward functions lead to equivalent optimal stopping problems; that is, we identify cases when they present the same surrender regions and optimal stopping times. 
While the idea of using an alternate continuous reward process to obtain continuity of the value function is not new to the literature, it is the first time, to the authors' knowledge, that an in-depth comparison of the two problems is presented.

The main contributions of this paper are summarized below.
\vspace{-5pt}
\begin{itemize}
	\item 
	We give a condition under which it is never optimal for the policyholder to surrender before maturity of the VA contract.
	\item We present an alternative representation for the value function of the VA contract in terms of a continuous reward function.  
	As stated above, this representation makes it easy to show continuity of the value function.
	In our paper, this alternative representation is used in Section \ref{sectionVAproperties} to show that the value function solves a free boundary value problem. These results justify the use of different numerical methods already applied to VA pricing in the actuarial literature.
	\item 
	We provide two further representations for the value function. 
	For the first one, we decompose the value of the contract into the value of the maturity benefit and an integral term akin to the early exercise premium from the American option literature.
	The second representation is new to the actuarial and American option pricing literature and expresses the contract value as the sum of the surrender benefit and an integral term coined the continuation premium.
	These two representations are used in Section \ref{sectSurrRegion} to characterize the {(non-) emptiness} of the surrender region. 
	The two representations can also be used to develop numerical pricing algorithms.
	\item We obtain conditions under which the original optimal stopping problem and the one with a continuous reward function are equivalent.
\end{itemize} 

This paper is organized as follows. Section \ref{sectFinSet} presents the optimal stopping problem involved in the pricing of VA contracts with a GMMB. The existence of an optimal stopping time is discussed in Section \ref{sectOptStop}.
In Section \ref{sectionVAproperties}, we obtain analytical properties of the value function, derive its integral representations and study the shape of the surrender region. We also discuss the equivalence between the original optimal stopping problem and its continuous reward function counterpart.
Numerical examples are provided in Section \ref{sectNumResults}.
Section \ref{sectionConclu} concludes the paper.

\section{Financial Setting}\label{sectFinSet}

\subsection{Market Model}

On a probability space $(\Omega,\mathcal{F},\mathbb{Q})$, let $W=\{W_t\}_{t\geq 0}$ be a standard Brownian motion whose augmented filtration is denoted by $\mathbb{F} = \{\mathcal F_t\}_{t\geq0}$. $\mathbb{Q}$ is the pricing (probability) measure in the market presented below.

We consider a financial market consisting of two primary assets, a risk-free bond $B=\{B_t\}_{t\geq 0}$ and a risky asset $S=\{S_t\}_{t\geq 0}$ whose dynamics under the measure $\mathbb{Q}$ are given by
\begin{equation}
	\begin{array}{rll}
		\diff B_t & = &rB_t\diff t,\\
		\diff S_t & =& r S_t \diff t+\sigma S_t \diff W_t,
	\end{array}\label{eqEDSBs}
\end{equation}
where $r$, $\sigma$ are deterministic constants with $r \geq 0$ and $\sigma>0$. It is easy to verify that \eqref{eqEDSBs} has a unique strong solution given by $B_t =  B_0 e^{rt}$  and
\begin{equation*}
	S_t  = S_0e^{\left(r-\frac{\sigma^2}{2}\right)t+\sigma W_t}. 
\end{equation*}
Readers will recognize the so-called Black-Scholes model, presented here directly in terms of its unique risk-neutral measure.
\subsection{Variable Annuity Contract}\label{subsectionVA}

We consider a simplified variable annuity contract offering a guaranteed minimum accumulation benefit at maturity $T\in\reals_+$, with $\reals_+=\{x\in\reals | x>0\}$. 
At inception of the contract, the policyholder deposits an initial premium in an investment account (or fund) tracking the risky asset. We denote by $F=\{F_t\}_{0\leq t\leq T}$ the value process of the account.

We remark that variable annuity policyholders usually benefit from a return-of-premium guarantee if they die before maturity. Under the assumption that mortality risk is completely diversifiable, it is straight-forward to add the death benefit in the analysis of the contract, see for example \cite{Mackay2017} and \cite{de2024variable}. In this paper, we ignore mortality risk.

The guarantees embedded in VA contracts are funded via a fee levied continuously from the investment account at a rate $c_t$ defined by
\begin{equation}\label{eqFeeFct}
	c_t:=c(t, F_t),\quad 0\leq t\leq T,
\end{equation}
where $c:[0,T]\times\reals_+\rightarrow [0,1]$, 
so that \eqref{eqEDSFs} has a unique strong solution. 
This fee structure is general enough to include state-dependent (\cite{bernard2014optimal}, \cite{delong2014pricing}, \cite{Mackay2017}) and time-dependent fees (\cite{bernard2019less}, \cite{kirkby2022valuation}).
The account value is then defined by
\begin{equation}
	F_t=S_te^{-\int_0^t c_u\diff u}, \quad 0\leq t\leq T, \label{eqF1}
\end{equation}
with $F_0=S_0$, so that
\begin{equation}
	\diff F_t  = (r-c_t) F_t \diff t+\sigma F_t \diff W_t.\label{eqEDSFs}
\end{equation}
Going forward, we denote by $\lbrace F_s^{t,x}\rbrace_{t\leq s\leq T}$ the solution to \eqref{eqEDSFs} with starting condition $F_t=x \in \reals_+$. When $t=0$, we simplify the notation and write $F_s^{x}=F_s^{0,x}$.

At maturity, the policyholder receives the maximum between a pre-determined amount $G\in\reals_+$ and the value of the investment account. 
Given $F_t=x$, the time-$t$ risk-neutral value of the maturity benefit is given by
\begin{equation}\label{eqVAnosurrender}
	h(t,x):=\e\left[e^{-r(T-t)}\max(G,F_T^{t,x})\right].
\end{equation}
Should she decide to surrender her contract prior to maturity, the policyholder receives the amount accumulated in the account subject to a penalty, expressed as a percentage that can depend on time and the value of the account. If no surrenders occur, the maturity benefit is paid at $T$.

More formally, let $\varphi:[0,\,T]\times \reals_+\rightarrow \reals_+$ denote the reward (or gain) function defined by
\begin{equation}
	\varphi(t,x)=\begin{cases}
		g(t,x)\,x & \text{if $t<T$,}\\
		\max(G,\,x) & \text{if $t=T$},
	\end{cases}\label{eqRewardFct}
\end{equation}
where $g:[0,\,T]\times\reals_+\rightarrow (0, 1]$ is in $C^{1,2}$ with bounded derivatives on $[0,T]\times\reals_+$, is non-decreasing in $t$ and satisfies $g(T,x)=1$ for all $x\in\reals_+$.
Thus, $\varphi(t,x)$ represents the amount received by the policyholder upon surrender at $t < T$ or at maturity $T$, given that the account has value $x$.

In the literature, $1-g$ is often called the surrender charge. 
It is non-increasing in time; a later surrender will yield a higher proportion of the account value.
In this paper, we use the term \emph{surrender charge function} to refer to either $g$ or $1-g$.

Imposing $g(T,\cdot)=1$ is common in the actuarial literature.
It allows the function $g$ to be defined on the closed interval $[0,T]$ and represents the fact that at maturity, the policyholder is entitled to the full amount accumulated in the account.
In practice, it is also common to see the surrender charge vanish before maturity. Examples of surrender charges are listed in \cite{palmer2006equity}. 

It is the first time, to our knowledge, that state-dependent surrender charges are considered. 
We will show that state-dependent surrender charges arise naturally when trying to mitigate surrender risk in the presence of a state-dependent fee.
\begin{remark}\label{rmkVarphiDiscoun}
	\begin{sloppypar}For $x<G$, the function $t\mapsto\varphi(t,x)$ is discontinuous at $T$ since 
		\begin{equation*}
			{\lim_{t\rightarrow T^-}\varphi(t,x)=g(T,x)x}= x<G=\varphi(T,x).
		\end{equation*}
	\end{sloppypar}
\end{remark}
Under the assumption that the policyholder maximizes the risk-neutral value of her contract, its value at time $t$ with account value $x$ is given by
\begin{equation}
	v(t,\, x)= \esssup_{\tau \in \mathcal{T}_{t,\, T}} \e[e^{-r(\tau -t)}\varphi(\tau, F_{\tau}^{t,\,x})]
	=\sup_{\tau \in \mathcal{T}_{t,\, T}} \e[e^{-r(\tau -t)}\varphi(\tau, F_{\tau}^{t,\,x})],\label{eqAmOptVA}  
\end{equation} 
for $(t,x) \in [0,T]\times \reals_+$, where $\mathcal{T}_{t,\, T}$ is the set of all stopping times taking value in the interval $[t,\,T]$. Moreover, throughout the paper, it must be understood that for $s > t$, 
\begin{equation*}
	v(s,F^{t,x}_s) =  \esssup_{\tau \in \mathcal{T}_{s,\, T}} \e[e^{-r(\tau -s)}\varphi(\tau, F_{\tau}^{t,\,x}) | \mathcal F_s],
\end{equation*}
since in that case, the essential supremum does not simplify.


\section{Optimal Stopping Time}\label{sectOptStop}

In this section, we show that under a simple condition on the fee and the surrender charge function, the optimal stopping problem in \eqref{eqAmOptVA} admits a trivial solution: an optimal stopping time is $T$, the maturity of the contract. 
In the second part of the present section, we discuss the existence of an optimal stopping time when this condition does not hold. 

We say that a stopping time $\tau^\star$ is \emph{optimal} if 
\begin{align}
	v(t,\, x)&= \sup_{\tau \in \mathcal{T}_{t,\, T}} \e[e^{-r(\tau-t) }\varphi(\tau, F_{\tau}^{t,x})]\label{eqDefOptStop}\\
	&= \e[e^{-r(\tau^\star-t)}\varphi(\tau^\star, F_{\tau^\star}^{t,x})]\nonumber.
\end{align}
In this case, we also say that $\tau^\star$ is \emph{optimal for \eqref{eqDefOptStop}}.

For a general reward function $\varphi$, such an optimal stopping time does not necessarily exist. 
However, under some regularity conditions on the discounted reward process ${Z=\{Z_t\}_{0\leq t\leq T}}$, with $Z_t:=e^{-rt}\varphi(t,F_t)$, the existence and the form of an optimal stopping time are well known from the theory of optimal stopping for random processes in continuous time, see, for instance, \cite{karatzas1998methods} (Appendix D, Theorem D.12),  \cite{peskir2006optimal} (Theorem 2.2), \cite{lamberton1998} (Section 10.2.1), \cite{Lamberton2009} (Theorem 2.3.5 and Section 2.3.4), \cite{el1981aspects}, Chapter 2. 
With the exception of \cite{el1981aspects}, these results rely on the (almost sure) continuity of the reward process, which does not hold for the one involved in variable annuity pricing since its trajectories are discontinuous at time $T$ with positive probability.
In fact the trajectories of $Z$ are upper semi-continuous.
Thus, additional work is required to show the existence of an optimal stopping time in our setting.

\subsection{Trivial Case}\label{subSectSurOptTrivial}
In this section, we derive a simple condition on the fee and surrender charge functions under which it is always optimal for the risk-neutral policyholder to hold the contract until maturity.

\begin{lemma}\label{lemmaVarphiSG}
	\begin{sloppypar}
		Let $g$ and $\varphi$ be defined as in \eqref{eqRewardFct}. If the discounted surrender value process $Y=\{Y_t\}_{0\leq t\leq T}$, with $Y_t:=e^{-rt}F_tg(t,\,F_t)$, is a submartingale, then the discounted reward process $Z = \{Z_t\}_{0\leq t\leq T}$, with $Z_t:=e^{-rt}\varphi(t, F_t)$, is also a submartingale.
	\end{sloppypar}
\end{lemma}

In the proof and throughout the paper, we use the notation $\e_t[\cdot] = \e[\cdot|\mathcal F_t]$.

\begin{proof}
	We first observe that for $0 \leq t < T$, $Y_t = Z_t$ so that if $Y$ is a submartingale, then for any $0\leq s\leq t<T$, $\e_s[Z_t] \geq Z_s$.
	For $0\leq s < t=T$, since $g(T,x)=1$ for all $x\in\reals_+$, we have
	\begin{equation*}
		\e_s[e^{-rT}\varphi(T,\,F_T)]
		\geq\e_s[e^{-rT}F_T g(T,\,F_T)]
		\geq e^{-rs}F_s g(s,\,F_s)
		=e^{-rs}\varphi(s,\,F_s).	
	\end{equation*}
\end{proof}
The next result is well known in optimal stopping theory (see, for example, \cite{bjork2009}, Proposition 21.2) and is reproduced here for completeness. 
It states that if the discounted reward process is a submartingale, maturity of the contract is an optimal stopping time. 
A simple financial interpretation is that if the discounted reward process $Z$ is a submartingale, it is non-decreasing on average over time. 
Holding the contract as long as possible is then optimal because its value is expected to increase. 
This is also the reasoning behind the optimal exercise strategy for an American call option on a non-dividend-paying stock. 
\begin{lemma} \label{lemmaOptStopTime}
	\begin{sloppypar}
		If the discounted reward process ${\{Z_t\}_{0 \leq t\leq T}}$, with $Z_t:=e^{-rt}\varphi(t, F_t)$, is a submartingale, then $T$ is an optimal stopping time, that is
		\begin{align*}
			v(t,\, x)&= \sup_{\tau \in \mathcal{T}_{t,\, T}} \e[e^{-r(\tau -t)}\varphi(\tau, F_{\tau}^{t,\,x})]
			=\e[e^{-r(T -t)}\varphi(T, F_{T}^{t,\,x})].
		\end{align*}
	\end{sloppypar}
\end{lemma}
Going forward, let $g_x(t,x)=\frac{\partial g}{\partial x}(t,x)$, $g_{xx}(t,x)=\frac{\partial^2 g}{\partial x^2}(t,x)$  and $g_t(t,x)=\frac{\partial g}{\partial t}(t,x)$.
\begin{proposition}\label{propEDPSurreder}
	Let $c$ and $g$ be the fee and the surrender charge functions as defined by \eqref{eqFeeFct} and \eqref{eqRewardFct}, respectively. If 
	\begin{equation}
		g_t(t,\,x)+(r-c(t,\,x)+\sigma^2)xg_{x}(t,\,x)+\frac{\sigma^2 x^2}{2}g_{xx}(t,\,x)-c(t,x)g(t,\,x)\geq 0\label{eqEDPg}
	\end{equation}
	for all $(t,\, x)\in [0,T)\times\reals_+$, $T$ is an optimal stopping time for \eqref{eqAmOptVA}.
\end{proposition}
\begin{proof}
	An application of Itô's lemma to the discounted surrender value process $Y$, defined in Lemma \ref{lemmaVarphiSG}, yields
	\begin{multline*}
		\diff Y_t=e^{-rt}F_t\Big(g_t(t,\,F_t)+(r-c(t,\,F_t)+\sigma^2)F_tg_{x}(t,\,F_t)+\frac{\sigma^2 F_t^2}{2}g_{xx}(t,\,F_t)\\
		-c(t,\,F_t)g(t,\,F_t)\Big)\diff t
		+e^{-rt}\sigma F_t\Big( g(t,F_t)+g_x(t,F_t)F_t\Big)\diff W_t.
	\end{multline*}
	Thus, if  
	\begin{equation*}
		g_t(t,\,x)+(r-c(t,\,x)+\sigma^2)xg_{x}(t,\,x)+\frac{\sigma^2 x^2}{2}g_{xx}(t,\,x)-c(t, x)g(t,\,x)\geq 0,
	\end{equation*}
	for all $(t,x)\in [0,T)\times\reals_+$, 
	$Y$  is a local submartingale. Furthermore, since $Y_t \leq e^{-rt} F_t \leq e^{-rt} S_t$, it follows from the properties of the geometric Brownian motion that for any $t < \infty$,
	\begin{align*}
		\sup_{\tau \leq t} \e[Y_\tau^2] \leq \sup_{\tau \leq t} \e[(e^{-r\tau}S_\tau)^2] < \infty,
	\end{align*}
	showing that $\{Y_\tau\}_{\tau \leq t}$ is uniformly integrable for any fixed $t < \infty$, and thus of class (DL) (see Definition \ref{defClassD}). This is sufficient to show that $Y$ is a true submartingale (see Proposition 1.7 of Chapter IV of \cite{revuz1990continuous}). 
	Then, by Lemma \ref{lemmaVarphiSG}, the discounted reward process ${\{Z_t\}_{0 \leq t\leq T}}$ is also a submartingale, and the optimal stopping time follows from Lemma \ref{lemmaOptStopTime}.
\end{proof}
\begin{remark} \label{rmkUniqunessT}
	In Section \ref{sectSurrRegion}, Theorem \ref{thmMainSurrenderRegion}\ref{thmMainSurrenderRegion2} shows that when the inequality in \eqref{eqEDPg} is satisfied for all $(t,x)\in[0,T)\times \reals_+$, the optimal stopping time $T$ is unique.
\end{remark}

To give a better idea of the impact of the fee and the surrender charge functions on the surrender incentives, we present some applications of Proposition \ref{propEDPSurreder}.
\begin{example}\label{ex1a}
	\begin{sloppypar}
		Let the fee rate be constant, that is, $c(t,x)=\bar{c}>0$, and let ${g(t,x)= e^{-k(T-t)}}$ for some $k>0$. To eliminate the incentive to surrender, it suffices to choose $k$ and $\bar{c}$ such that \eqref{eqEDPg} holds; a quick calculation yields $k \geq \bar{c}$.
		In this setting, it suffices to set $k$ at or above the fee rate $\bar{c}$ to eliminate optimal early surrenders. This result is well-known and discussed in \cite{BernardMackay2015}, Proposition 3.1 and \cite{Mackay2014}, Proposition 4.4.2. In particular, \cite{Mackay2014} shows that in the Black-Scholes setting, $g(t,x)= e^{-\bar{c}(T-t)}$ is the minimal surrender charge function that can be used to remove the surrender incentive when the fee rate is constant. 
	\end{sloppypar}
\end{example}
\begin{example}\label{ex1b}
	Let $g(t,x)=1- k\left(1-\frac{t}{T}\right)^3$ for some constant $k>0$,  as in \cite{bacinello2019variable}, \cite{Mackay2014} Section 4.3.2, and \cite{Mackay2017}. As explained in \cite{Mackay2014}, this form of surrender charge function mimics the surrender charges in the market, which are usually high in the first years of the contract and drop drastically thereafter. Assume also that the fee function depends only on time such that $c(t,x)=c(t)$ for all $x \in \reals_+$. In this example, we do not specify a particular form for the fee function; instead, \eqref{eqEDPg} is used to obtain the form $c(t)$ eliminating the surrender incentive. Simple calculations show that \eqref{eqEDPg} is satisfied if
	\begin{equation}
		c(t)\leq \frac{\frac{3k}{T}\left(1-t/T\right)^2}{1-k(1-t/T)^3}.\label{eqEx1bCt}
	\end{equation}
	Thus, any function $c(t)$ that satisfies this inequality will eliminate surrender incentives. For example, one could simply set $c(t)=\frac{3 k /T\left(1-t/T\right)^2}{1-k (1-t/T)^3}$.
	This example highlights the interplay between the fee and the surrender charge structures. In this setting, it would have been impossible to satisfy \eqref{eqEDPg} using a constant fee rate, except with $c=0$, as the function on the right-hand side of \eqref{eqEx1bCt} vanishes at maturity. Having a zero fee rate throughout the life of the VA contract is not a feasible solution from the insurer's perspective since the maturity guarantee must be financed. On the other hand, isolating $k$ in \eqref{eqEx1bCt} gives
	\begin{equation}
		k\geq \frac{c(t)}{3/T\left(1-t/T\right)^2+c(t)(1-t/T)^3}.\label{eqEx1bCt2}
	\end{equation}
	The right-hand side of the inequality goes to infinity as $t$ approaches $T$ if $c(t)$ is kept constant, which gives another reason to allow the fee to vary over time.
\end{example}
The examples above show that in order to efficiently eliminate early surrender incentives, surrender charges and fees should be structured conjointly.
\begin{remark}
	The results of Section \ref{subSectSurOptTrivial} can easily be extended to more general market models and fee structures, such as the VIX-linked fee discussed in \cite{cui2017}, \cite{kouritzin2018vix}, and \cite{MackayVachonCui2022VaCTMC}.
\end{remark}

\subsection{General Case}
In this section, we study the optimal surrender strategy for the variable annuity contract in the general case, when condition \eqref{eqEDPg} is not necessarily satisfied.
We first recall two definitions from probability and optimal stopping theory. 

\begin{defi}
	A process $X=\lbrace X_t\rbrace_{t\geq 0}$ defined on a filtered probability space  $\lbrace \Omega, \mathcal{F},\lbrace \mathcal{F}_t\rbrace_{t\geq 0},\p\rbrace$ is said to be \textit{optional} if it is measurable with respect to the sigma-algebra generated by the right-continuous and adapted processes.
\end{defi}

\begin{defi}\label{defClassD}
	A right-continuous adapted process $X=\lbrace X_t\rbrace_{t\geq 0}$ defined on a filtered probability space  $\lbrace \Omega, \mathcal{F},\lbrace \mathcal{F}_t\rbrace_{t\geq 0},\p\rbrace$ is said to be 
	\begin{itemize}
		\item of \textit{class (D)} if $\lbrace X_{\tau}, \tau\in\mathcal{T}_{0,\infty}\rbrace$ is uniformly integrable;
		\item of \textit{class (DL)} if for each $t \geq 0$, $\lbrace X_{\tau}, \tau\in\mathcal{T}_{0,t}\rbrace$ is uniformly integrable.
	\end{itemize}
\end{defi}
The next theorem is from Theorem 19 of \cite{bassan2002optimal}, which summarizes concisely the results of \cite{el1981aspects}. An advised reader will notice differences between the conditions stated in Theorem 19 of \cite{bassan2002optimal} and the ones of Theorem \ref{thmElkaroui81}. This is because \cite{bassan2002optimal} work with bounded reward functions so the class (D) condition of \cite{el1981aspects} is automatically satisfied in their setting.  
\begin{theorem}[\cite{el1981aspects}, Theorems 2.28, 2.31 and 2.41, see also \cite{bassan2002optimal}, Theorem 19]\label{thmElkaroui81}
	Let the discounted reward process $\lbrace Z_t\rbrace_{t\geq 0}$ defined on some filtered probability space  $\lbrace \Omega, \mathcal{F},\lbrace \mathcal{F}_t\rbrace_{t\geq 0},\p\rbrace$ be optional, non-negative and of class (D),
	and let $\lbrace J_t\rbrace_{ t\geq 0}$ denote its Snell envelope, that is
	\begin{equation*}
		J_t=\esssup_{\tau\geq t} \e[Z_{\tau}|\mathcal{F}_t].
	\end{equation*}
	Then,
	\begin{enumerate}[label=(\roman*)]
		\item $J$ is the smallest non-negative supermartingale which dominates $Z$ (\cite{el1981aspects}, Theorem 2.28). \label{enumElKarouiSM}
		\item A stopping time $\tau^\star\in\mathcal{T}_{0,\,\infty}$ is optimal if and only if $J_{\tau^\star}=Z_{\tau^\star}$ a.s. and $\lbrace J_{\tau^\star\land t}\rbrace_{t\geq 0}$ is a martingale (\cite{el1981aspects}, Theorem 2.31).\label{enumElKarouiMartingale}
		\item If the trajectories of $Z$ are upper semi-continuous, then 
		\begin{equation*}
			\tau^\star=\inf\lbrace t\geq 0 |J_t=Z_t\rbrace,
		\end{equation*}
		is the smallest optimal stopping time (\cite{el1981aspects}, Theorem 2.41).\label{enumElKarouiMinTau}
	\end{enumerate}
\end{theorem}

In light of Theorem \ref{thmElkaroui81}, since our discounted reward process is positive and adapted with upper semi-continuous trajectories, showing the existence of an optimal stopping time is simply a matter of verifying that $Z$ is of class (D).
\begin{corollary}\label{corrOptStop1} Let $\varphi$ and $v$ be defined by \eqref{eqRewardFct} and \eqref{eqAmOptVA}, respectively. The stopping time given by
	\begin{equation}
		\tau^x_t=\inf\left\lbrace t \leq s\leq T \big|\,v(s, F_s^{t,\,x})= \varphi(s, F_s^{t,\,x})\right\rbrace.\label{eqOtpStop1}
	\end{equation}
	is optimal for \eqref{eqAmOptVA}.
\end{corollary}
\begin{proof}
	Since $Z = \{Z_t\}_{0\leq t \leq T}$, with $Z_{t}=e^{-rt}\varphi(t,F_{t})$, is dominated by an integrable non-negative random variable, $\sup_{0\leq t\leq T} e^{-rt}\varphi(t,F_t)$, it is of class (D). The result then follows from Theorem \ref{thmElkaroui81}.
\end{proof}

Theorem \ref{thmContRep}, inspired by \cite{palczewski2010finite} and presented below, provides an alternative representation for the optimal stopping problem in \eqref{eqAmOptVA} and will be essential to prove the continuity of the value function.

\begin{theorem}\label{thmContRep}
	The value function $v$ defined in \eqref{eqAmOptVA} can be written as
	\begin{equation}\label{eqVASurOpt3}
		v(t,x)=\sup_{\tau\in\mathcal{T}_{t,\, T}}\e\left[e^{-r(\tau-t)}\left(g(\tau,F_{\tau}^{t,x})F_{\tau}^{t,x} \vee h(\tau,\, F_{\tau}^{t,x})\right) \right],	
	\end{equation}
	where
	\begin{equation*}
		h(t,\,x)=\e\left[e^{-r(T-t)}\max(G,F_{T}^{t,x})\right], \quad t\in[0,T].
	\end{equation*}
\end{theorem}

\begin{proof}
	Without loss of generality, let $t=0$ and note that for $0\leq s \leq T$,
	\begin{equation*}
		h(s,F_s)=\e\left[e^{-r(T-s)}\max(G,F_{T})|\mathcal{F}_s\right].
	\end{equation*}
	
	We define the Snell envelope of the discounted continuous reward process $\{\tilde J_s\}_{0\leq s \leq T}$ by
	\begin{align*}
		\tilde J_s = \esssup_{\tau \in \mathcal T_{s,T}} 
		\e[{e^{-r\tau}g(\tau,F_\tau)F_\tau \vee h(\tau,F_\tau)\mid \mathcal F_s}]
	\end{align*}
	for $0 \leq s \leq T$, and recall that the Snell envelope of the discounted original reward process $\{J_s\}_{0\leq s\leq T}$ is defined by 
	\begin{align*}
		J_s = \esssup_{\tau \in \mathcal T_{s,T}} 
		\e\left[e^{-r\tau}\varphi(\tau,F_\tau)\mid \mathcal F_s\right].
	\end{align*}
	We show that $\tilde J_s = J_s$ a.s. for all $0 \leq s \leq T$.
	
	First observe that since $g(t,x)x \vee h(t,x) \geq \varphi(t,x)$ for all $(t,x) \in [0,T]\times\reals_+$, $\tilde J_s \geq J_s$ for all $0 \leq s \leq T$.
	
	Now we show that  $\tilde J_s \leq J_s$ for all $0 \leq s \leq T$.
	By definition of $v$, $g(t,x)x \vee h(t,x) \leq v(t,x)$ for all $(t,x) \in [0,T]\times\reals_+$. 
	It follows that
	\begin{align*}
		\tilde J_s 
		\leq 
		\esssup_{\tau \in \mathcal T_{s,T}} \e\left[e^{-r\tau} v(\tau, F_\tau) \mid \mathcal F_s\right]
		= \esssup_{\tau \in \mathcal T_{s,T}} \e\left[J_\tau \mid \mathcal F_s\right].
	\end{align*}
	Since $\{J_s\}_{0\leq s \leq T}$ is a supermartingale, by the optimal sampling theorem, $\e[J_\tau \mid \mathcal F_s] \leq J_s$ for any $\tau \in \mathcal T_{s,T}$, and thus 
	$\esssup_{\tau \in \mathcal T_{s,T}} \e\left[ J_\tau \mid\mathcal F_s\right] \leq J_s$,
	which completes the proof.

	\end{proof}
\begin{remark}\label{rmkContinuityModifiedReward}
	The modified reward function $g(t,x)x\vee h(t,\,x)$ is continuous since it is the maximum of two continuous functions. 
	Indeed, $g(t,x)x$ is continuous by definition while the continuity of $h(t,\,x)$ follows from Theorem 3 of \cite{veretennikov1981strong}, which is the analog of the Feynman-Kac Theorem (see, for instance, \cite{karatzasShreve1991}, Theorem 5.7.6). 
\end{remark}
Similarly to \cite{palczewski2010finite}, we use the new representation of the value function in \eqref{eqVASurOpt3} to construct another optimal stopping time for the original problem in \eqref{eqAmOptVA}. Define 
\begin{equation}
	\tilde{\tau}_t^x:=\inf\left\lbrace t\leq s\leq T |v(s,F_s^{t,x})= g(s,F_{s}^{t,x})F_{s}^{t,x} \vee h(s,\, F_{s}^{t,x})\right\rbrace, \label{eqOtpStop2}
\end{equation}
the smallest optimal stopping time for \eqref{eqVASurOpt3}, as per Theorem \ref{thmElkaroui81}\ref{enumElKarouiMinTau}.
\begin{corollary} \label{corrOptStop}
	Let $h$ and $g$ be defined as in \eqref{eqVAnosurrender} and \eqref{eqRewardFct}, respectively. The stopping time defined by
	\begin{equation}
		\bar{\tau}^x_t=\begin{cases}
			\tilde{\tau}^x_t, & g(\tilde{\tau}^x_t,F_{\tilde{\tau}^x_t}^{t,x})F_{\tilde{\tau}^x_t}^{t,x}\geq h(\tilde{\tau}^x_t,F_{\tilde{\tau}^x_t}^{t,x}),\\
			T, & g(\tilde{\tau}^x_t,F_{\tilde{\tau}^x_t}^{t,x})F_{\tilde{\tau}^x_t}^{t,x}< h(\tilde{\tau}^x_t,F_{\tilde{\tau}^x_t}^{t,x}), \label{eqOtpStop3}
		\end{cases}
	\end{equation}
	is optimal for \eqref{eqAmOptVA}.
\end{corollary}
\begin{proof}
	To keep the notation lighter, 
	we write $\bar{\tau}=\bar{\tau}^{x}_t$ and $\tilde{\tau}=\tilde{\tau}^{x}_t$. By Theorems \ref{thmElkaroui81} and \ref{thmContRep}, we have
	\begin{align*}	
		v(t,x)
		&=\e\left[e^{-r(\tilde{\tau}-t)}\left( g(\tilde{\tau},F_{\tilde{\tau}}^{t,x})F_{\tilde{\tau}}^{t,x} \vee h(\tilde{\tau},\, F_{\tilde{\tau}}^{t,x})\right) \right]\\
		&=\e\left[e^{-r(\tilde{\tau}-t)}g(\tilde{\tau},F_{\tilde{\tau}}^{t,x})F_{\tilde{\tau}}^{t,x}\ind_{\lbrace  g(\tilde{\tau},F_{\tilde{\tau}}^{t,x})F_{\tilde{\tau}}^{t,x}\geq h(\tilde{\tau},F_{\tilde{\tau}}^{t,x}) \rbrace} \right.\\
		& \qquad +\left. e^{-r(\tilde{\tau}-t)}h(\tilde{\tau},\, F_{\tilde{\tau}}^{t,x})\ind_{\lbrace h(\tilde{\tau}, F_{\tilde{\tau}}^{t,x})> g(\tilde{\tau},F_{\tilde{\tau}}^{t,x})F_{\tilde{\tau}}^{t,x} \rbrace }\right]\\
		&=\e\left[e^{-r(\tilde{\tau}-t)}g(\tilde{\tau},F_{\tilde{\tau}}^{t,x})F_{\tilde{\tau}}^{t,x}\ind_{\lbrace  g(\tilde{\tau},F_{\tilde{\tau}}^{t,x})F_{\tilde{\tau}}^{t,x}\geq h(\tilde{\tau},F_{\tilde{\tau}}^{t,x}) \rbrace} \right]\\
		& \qquad +\e\left[\e\left[e^{-r(T-t)}\max(G,F_{T}^{t,x})\ind_{ \lbrace h(\tilde{\tau}, F_{\tilde{\tau}}^{t,x})> g(\tilde{\tau},F_{\tilde{\tau}}^{t,x})F_{\tilde{\tau}}^{t,x} \rbrace}\Big|\mathcal{F}_{\tilde{\tau}}\right]\right]\\
		&=	\e\left[e^{-r(\tilde{\tau}-t)}g(\tilde{\tau},F_{\tilde{\tau}}^{t,x})F_{\tilde{\tau}}^{t,x}\ind_{\lbrace  g(\tilde{\tau},F_{\tilde{\tau}}^{t,x})F_{\tilde{\tau}}^{t,x}\geq h(\tilde{\tau},F_{\tilde{\tau}}^{t,x}) \rbrace} \right]\\
		& \qquad +\e\left[e^{-r(T-t)}\max(G,F_{T}^{t,x})\ind_{ \lbrace h(\tilde{\tau}, F_{\tilde{\tau}}^{t,x})> g(\tilde{\tau},F_{\tilde{\tau}}^{t,x})F_{\tilde{\tau}}^{t,x} \rbrace}\right]\\
		&=\e\left[e^{-r(\bar{\tau}-t)}\varphi(\bar{\tau}, F_{\bar{\tau}}^{t,x})\right].
	\end{align*}
	\end{proof}
It is possible to show that $\bar{\tau}_t ^{x} =\tau_t^{x}$ a.s., where $\tau_t^x$ is the smallest optimal stopping time for the original problem defined in Corollary \ref{corrOptStop1}. 
A comparison of the three stopping times is given below.
\begin{lemma}\label{lemmaStoppingTimes}
	For $0 \leq t \leq T$, the optimal stopping times defined in \eqref{eqOtpStop1}, \eqref{eqOtpStop2} and \eqref{eqOtpStop3}  satisfy $\tilde{\tau}_t^{x}\leq \tau_t^{x}$ a.s. and $\tau_t^{x}=\bar{\tau}_t^{x}$ a.s.
\end{lemma}
\begin{proof}
	For $t=T$, the result trivially holds with equality.
	Henceforth, let $t < T$.
	To show $\tilde{\tau}_t^{x}\leq \tau_t^{x}$, let $F_t=x$ and consider three cases.
	
	\begin{sloppypar}
		(1) Suppose ${v(t,x)>  g(t,x)x \vee h(t, x) \geq \varphi(t,x)}$. 
		Then, for any $s\in[t,T]$, ${g(s,F_{s}^{t,x})F_{s}^{t,x} \vee h(s,\, F_{s}^{t,x})\geq \varphi(s, F_{s}^{t,x})}$ a.s., so, the process $v(s, F_s)$ must first cross the continuous reward process $g(s, F_s)F_s \vee h(s, F_s)$ to attain the discontinuous reward process $\varphi(s, F_s)$ (since the process $v(s, F_s)$ starts above the two reward processes at $t$). It follows that $\tilde{\tau}_t^{x}\leq \tau_t^{x}$. (2) If $v(t,x)=h(t,x)$ then necessarily $h(t,x)\geq xg(t,x)$, since by the continuous reward representation $v(t,x)\geq h(t,x)\vee xg(t,x)$,  so that $\tilde{\tau}_t^{x}=t$, and $\tau_t^{x}\geq t$, and the first inequality holds. The last case (3) is when $v(t,x)=xg(t,x)$, which automatically implies $\tilde{\tau}_t^{x}=\tau_t^{x}= t$.
	\end{sloppypar}
	
	
	\begin{sloppypar}
		To show that \(\tau_t^{x}=\bar{\tau}_t^{x}\), we fix \(\omega \in \Omega\) and consider two cases: (1) \(\tau_t^{x}(\omega) < T\) and (2) \(\tau_t^{x}(\omega) = T\).  
		
		(1) By the definition of \(\tau_t^{x}\) and the reward function \(\varphi\), we have
		${v(\tau_t^{x}(\omega),F^{t,x}_{\tau_t^{x}(\omega)}(\omega)) = g(\tau_t^{x}(\omega),F^{t,x}_{\tau_t^{x}(\omega)}(\omega))F^{t,x}_{\tau_t^{x}(\omega)}(\omega)}$, and if \(\tilde{\tau}_t^x(\omega)= \tau_t^{x}(\omega)\), it follows that  ${\tilde{\tau}_t^x(\omega)= \tau_t^{x}(\omega)= \bar{\tau}_t^x(\omega).}$ From the definition of \(\tau_t^{x}(\omega)\), we know that
		${v(s,F^{t,x}_{s}(\omega)) > g(s,F^{t,x}_s(\omega))F^{t,x}_s(\omega)}$ for all ${t\leq s < \tau^x_t(\omega).}$
		Thus, for \(\tilde{\tau}_t^x(\omega)\) to be equal to \(\tau_t^x(\omega)\), we need
		$
		{v(s,F^{t,x}_{s}(\omega)) > h(s,F^{t,x}_s(\omega))}$, for all  ${t\leq s < \tau^x_t(\omega)}.
		$
		Now, suppose there exists some \(s\) in the interval \([t, \tau_t^x(\omega))\) such that
		$
		{v(s,F^{t,x}_{s}(\omega))= h(s,F^{t,x}_s(\omega))}.
		$
		In this case, the two stopping times \(\tau_t^x(\omega)\) and \(\bar{\tau}_t^x(\omega)\) would differ. However, the equality \({v(s,F^{t,x}_{s}(\omega))= h(s,F^{t,x}_s(\omega))}\) would imply that the probability that the contract is exercised before maturity is nil, which contradicts our initial assumption that \(\tau_t^x(\omega)< T\).  
		Therefore, it must hold that
		$
		{v(s,F^{t,x}_{s}(\omega)) > h(s,F^{t,x}_s(\omega))}$ for all ${t\leq s < \tau^x_t(\omega)}$,
		which leads to 
		$
		\tilde{\tau}_t^x(\omega)=\tau_t^{x}(\omega)=\bar{\tau}_t^{x}(\omega).
		$
	\end{sloppypar}
	(2) $\tau_t^{x}(\omega) = T$ implies $v(s,F^{t,x}_s(\omega)) > g(s, F^{t,x}_s(\omega)) F^{t,x}_s(\omega)$
	for all $t \leq s \leq T$. This means that if $\tilde{\tau}_t^{x}(\omega) < T$, $v(\tilde{\tau}_t^{x}(\omega),F_{\tilde{\tau}_t^{x}(\omega)}(\omega)) = h(\tilde{\tau}_t^{x}(\omega),F_{\tilde{\tau}_t^{x}(\omega)}(\omega))$ and by definition of $\bar{\tau}_t^{x}$, $\bar{\tau}_t^{x}(\omega) = \tau_t^{x}(\omega)$. 
	\end{proof} 

The intuition behind the equality between $\tau^x_t$ and $\bar \tau^x_t$ is that $\bar \tau^x_t$ depends on \emph{why} the contract is surrendered in the problem with the continuous reward. 
Indeed, the strategy underlying $\tilde \tau^x_t$ is to stop if the contract value is equal to either the surrender benefit or the expected present value (EPV) of the maturity benefit. 
If the reason to stop is equality with the surrender benefit, then this optimal stopping time coincides with the strategy underlying the problem with the discontinuous reward function and all three stopping times are equal.
However, if the reason to stop is that the value of the contract is equal to the EPV of the maturity benefit, then $\bar\tau^x_t$ is set to maturity $T$.
In that case, it is also optimal to wait until maturity in the original problem, since the contract being worth the EPV of the maturity benefit means that $T$ is the optimal stopping time.

A simple example involving these optimal stopping times is when condition \eqref{eqEDPg} of Proposition \ref{propEDPSurreder} is satisfied.
\begin{example}\label{exOptStopTrivialCase}
	Suppose that the fee and the surrender charge functions defined in \eqref{eqFeeFct} and \eqref{eqRewardFct}, respectively, satisfy \eqref{eqEDPg}. Hence, by Proposition \ref{propEDPSurreder} and Remark \ref{rmkUniqunessT}, the unique optimal stopping time for the problem with the discontinuous reward function is $\tau_0^x=T$, so that
	\begin{equation*}
		v(0,x)=\e[e^{-rT}\max(G,F_T^{x})]=h(0,x).
	\end{equation*}
	Clearly, an optimal stopping time for the problem with the continuous reward function is $\tilde{\tau}_0^x=0$ since $h(0,x)=\e[e^{-rT}\max(G,F_T^x)]=v(0,x)$; whereas $\bar{\tau}_0^x=T$ as per  \eqref{eqOtpStop3}, since $v(0,x)>g(0,x)x$. 
\end{example}
In the next section, we present a condition on the fee and surrender charge functions under which $\tilde{\tau}_t^x=\tau_t^x$ a.s. 
This result is particularly interesting since it provides a condition under which the two optimal stopping problems defined in \eqref{eqAmOptVA} and \eqref{eqVASurOpt3} are equivalent. 
That is, the two problems have the same value function (Theorem \ref{thmContRep}), surrender region (Proposition \ref{propSurrRegionsEqual}), and optimal stopping strategy (Corollary \ref{corrEqualOptStop}). 


	\section{Analytical study of the value function} \label{sectionVAproperties}

The representation of the value function in terms of a continuous reward function simplifies the analysis of its regularity, which is presented in Section \ref{sectionPropv}.
However, additional work is required to assess its smoothness, which we discuss in Proposition \ref{propLipchtitzV} below.
In Section \ref{subSectionVI}, we establish the relationship between $v$, a free boundary value problem, and a variational inequality. 
This allows us to derive, in Section \ref{sectionSurOptV1}, two other representations for the value function: the surrender premium representation which is analogous to the exercise premium representation in the American option pricing terminology, and the continuation premium representation. 
This representation is new to the literature on VA and American option pricing and leads to the characterization of the (non-)emptiness of the surrender region in Section \ref{sectSurrRegion}. 
Section \ref{sectionEquivalenceProblem} presents a condition under which the optimal stopping problem with discontinuous reward function defined in \eqref{eqAmOptVA} is equivalent to the one with the continuous reward function, in \eqref{eqVASurOpt3}.

Assumptions \ref{assumpCHolderContinuous} and \ref{assumpCandg} below are needed for several of the results of this section to hold. 
In Assumption \ref{assumpCHolderContinuous} below, ${\mu:[0,T]\times\reals_+\rightarrow \reals}$ denotes the drift term of the account value process \eqref{eqEDSFs}, that is, $\mu(t,x)=(r-c(t,x))x$.

\begin{assump}\label{assumpCHolderContinuous}
	\begin{enumerate}[label=(\roman*)] 
		\item
		The fee function $c$ is such that $\mu$ is continuous and globally Lipschitz in $x$, that is, 
		there exists $K\geq 0$ such that for all $x,y\in\reals_+$, $t\geq 0$,
		\begin{equation*}
			|\mu(t,x)-\mu(t,y)|\leq K |x-y|.
		\end{equation*}
		\label{enumCHolderContinuouspart1}
		\item
		For all $x$, the function $t \mapsto c(t,x)$ is locally H\"{o}lder continuous.
	\end{enumerate}		
\end{assump}

\begin{assump}\label{assumpCandg}
	The fee and the surrender charge functions only depend on time; that is, we suppose that $c(t,x)=c(t)$ and $g(t,x)=g(t)$  for all $x\in\reals_+$. 
\end{assump}
\begin{remark}
	As soon as Assumption \ref{assumpCandg} holds, the drift term in \eqref{eqEDSFs} is globally Lipschitz in $x$, so Assumption \ref{assumpCHolderContinuous}\ref{enumCHolderContinuouspart1} is automatically satisfied.
	However, we still write the two assumptions separately, since for some of the results in this section, the assumptions can be relaxed and only one of the two assumptions is needed.
\end{remark}

For the rest of this paper, we assume that Assumptions \ref{assumpCHolderContinuous} and \ref{assumpCandg} hold. 
This restricts the applicability of some of our results.
For example, Assumptions \ref{assumpCHolderContinuous}\ref{enumCHolderContinuouspart1} and \ref{assumpCandg}, preclude the use of the state-dependent fee function of \cite{Bernard2014}.
Not all of our results require both assumptions to hold; when this is the case, it will be indicated in a remark after the result.

\subsection{Elementary Properties of the Value Function}\label{sectionPropv}

An important result for the analysis of the value function is its continuity, which is easily obtained from its alternative representation in terms of a continuous reward function.

\begin{theorem}\label{thmContinuity_v}
	The value function $v$ is continuous on $[0,T]\times\reals_+$.
\end{theorem}

\begin{proof}
	The continuity of the reward function implies the continuity of the value function (see \cite{krylov2008controlled}, Theorem 3.1.8, which requires Lipschitz continuity of the drift term in \eqref{eqEDSFs}, hence the need for Assumption \ref{assumpCHolderContinuous}\ref{enumCHolderContinuouspart1}). Thus, 
	the assertion follows from the continuous reward representation of the value function in Theorem \ref{thmContRep} and Remark \ref{rmkContinuityModifiedReward}. 
\end{proof}

\begin{remark}
	Theorem \ref{thmContinuity_v} only requires Assumption \ref{assumpCHolderContinuous} \ref{enumCHolderContinuouspart1} to hold.
	In particular, the value function is continuous even when the fee and surrender functions depend on both time and the account value, as long as the drift term of \eqref{eqEDSFs} is globally Lipschitz in $x$.
\end{remark}

Some basic properties of the value function, such as local boundedness, are derived in the next lemma.
\begin{lemma}\label{lemmaVprop1} For every $(t,\,x)\in [0,\, T]\times \reals_+$, the value function $v$ satisfies the following properties: 
	\begin{enumerate}[label=(\roman*)]
		\item $Ge^{-r(T-t)}\leq v(t,\,x)\leq G+x$; \label{enumVprop1}
		\item $v(t,x)\geq \varphi(t,x)$; \label{enumVprop2} 
		\item $v(T,x)=\max(G,x)$. \label{enumVprop3}
	\end{enumerate}
\end{lemma}
The first assertion follows from optional sampling since $\{e^{-r(u-t)} F_u^{t,x}\}_{u=t}^T$ is a supermartingale. Assertions \ref{enumVprop2} and \ref{enumVprop3} follow easily from \eqref{eqAmOptVA}.

\begin{remark}
	The properties given in Lemma \ref{lemmaVprop1} follow from the definition of the contract. They do not require Assumptions \ref{assumpCHolderContinuous} and \ref{assumpCandg} to be satisfied.
\end{remark}

\begin{lemma}\label{lemmaVprop2}
	The function $x\mapsto v(t,\,x)$ is non-decreasing and convex for all $t\in [0,T]$.
\end{lemma}
The proof is a direct consequence of the convexity and the non-decreasing property of $x\mapsto\varphi(t,\,x)$, for each $t\in [0,\,T]$.

\begin{remark}
	Only Assumption \ref{assumpCandg} is necessary for Lemma \ref{lemmaVprop2} to hold.
\end{remark}

\begin{remark}\label{rmkNonMonotonicity_t}
	The value function in the American option pricing problem is usually non-increasing in time. This will always be the case if the underlying asset price process and the reward function are time-homogeneous, which is true for call and put options under the Black-Scholes model. 
	However, in our setting, the reward function is time-dependent and may be increasing in $t$ (because of the surrender charge function $g$). Thus, the value function is not necessarily monotone in time. 
\end{remark}

Proposition \ref{propLipchtitzV} below concerns the smoothness of the value function. 
Our proof, which can be found in the Appendix, presents non-trivial differences compared to the American put case, which are due to the time-dependence of the value function. 

\begin{proposition}$\quad$
	\begin{enumerate}[label=(\roman*)]
		\item 
		For every $t\in [0,\,T]$, and for $x,\,y\geq 0$, \label{enumPropLip_v_i}
		\begin{equation*}
			|v(t,x)-v(t,\,y)|\leq |x-y|. 
		\end{equation*}
		\item 
		For every $x\in\reals_+$, there exist constants $C_1,\,C_2>0$ (which may depend on $x$) satisfying,
		\begin{equation*}
			|v(t,x)-v(s,x)|\leq C_1(t-s)+C_2\sqrt{t-s},
		\end{equation*}
		
		for  $0 \leq s < t \leq T$.\label{enumPropLip_v_ii}
	\end{enumerate}
\end{proposition}\label{propLipchtitzV}

\begin{remark}
	Proposition \ref{propLipchtitzV} holds in a more general setting.
	Part \ref{enumPropLip_v_i} requires that the fee function depends only on time and that the surrender charge function is such that $x \mapsto xg(t,x)$ remains Lipschitz continuous. 
	Part \ref{enumPropLip_v_ii} requires that Assumption \ref{assumpCandg} is satisfied.
\end{remark}
%


\subsection{Free Boundary Value Problem and Variational Inequality}\label{subSectionVI}

The continuous reward representation of Theorem \ref{thmContRep} and the previously established continuity of the value function allows us to apply classical results from optimal stopping theory.
In this section, we use this strategy to establish the relationship between the value function $v$ associated with the original optimal stopping problem \eqref{eqAmOptVA}, a free boundary value problem, and a variational inequality.
Thus, the proofs of the results presented in this section are kept short; the reader is referred to the literature for more details.

The value function is supported on $[0,T]\times\reals_+$, which is separated into  the \emph{surrender region} $\Sr$ and the \emph{continuation region} $\mathcal{C}$, defined by
\begin{equation}
	\mathcal{S}=\left\{(t,x)\in [0,T)\times\reals_+ | v(t,x)=\varphi(t,x)\right\},\text{ and}\label{eqSurrRegionDisc}
\end{equation}
\begin{equation}\label{eqContinuationRegionDisc}
	\mathcal{C}=\lbrace (t,x)\in [0,\,T)\times \reals_+:v(t,x)>\varphi(t,x)\rbrace,
\end{equation}
respectively. It follows that $\Cr\cup\Sr=[0,T)\times\reals_+$, since $v(t,x)\geq \varphi(t,x)$ for all $(t,x)\in [0,T]\times \reals_+$.
Furthermore, since $v$ and $\varphi$  are continuous on $[0,T)\times\reals_+$ (see Theorem \ref{thmContinuity_v}), the continuation region is open and the surrender region is closed. The exercise (or surrender) boundary is the boundary $\partial\Cr$ of $\Cr$.

To express the value function as the solution of a free boundary problem, we define the second-order differential operator $\mathcal{L}_t$, for $0 \leq t \leq T$, by
\begin{equation*}
	\mathcal{L}_t:=\frac{x^2\sigma^2}{2}\frac{\partial^2}{\partial x^2}+(r-c(t,x))x\frac{\partial}{\partial x},  
\end{equation*} 
and the function ${L:[0,T)\times\reals_+\rightarrow \reals}$ by
\begin{equation}
	\begin{split}\label{eqDefinitionL}
		L(t,x)&:= \frac{\mathcal{L}_t(x g(t,x))}{x}+g_t(t,x)-rg(t,x)\\
		&= g_t(t,\,x)+(r-c(t,\,x)+\sigma^2)xg_{x}(t,\,x)\\
		&\qquad+\frac{\sigma^2 x^2}{2}g_{xx}(t,\,x)-c(t,x)g(t,\,x).
	\end{split}
\end{equation}
Note that $L$ is the term on the left-hand side of  \eqref{eqEDPg}.  
We give this general definition of $L$ since many of our results apply in a general setting, when the fee and surrender charge function depend on both $t$ and $x$.
However, under Assumption \ref{assumpCandg}, which is assumed to hold throughout Section \ref{sectionVAproperties}, $L$ becomes 
\begin{equation}\label{eqDefinitionLt}
	L(t)= g_t(t)-c(t)g(t).
\end{equation}

The next results confirm that $v(t,x)$ can be expressed as the solution to a free boundary problem.

\begin{theorem}\label{thmFBproblemGen}
	The value function $v$ solves the boundary value problem
	\begin{equation}
		\left\lbrace
		\begin{aligned}
			\mathcal{L}_tf(t,x)+f_t (t,x)-rf(t,x) & =0, &\quad (t,x)\in\mathcal{C}\\ 
			f(t,x) &>xg(t,x), & \quad (t,x)\in \Cr\\
			f(t,\,x)& =xg(t,x), &\quad (t,x)\in \Sr\\
			f(T,x) & =\max(G,x), & x\in\reals_+\\
		\end{aligned}\label{eqFBprobGen}
		\right.
	\end{equation}
	In particular, the function $v \in C^{1,2}$ on $\Cr$.
\end{theorem}

This result follows from the theory on optimal stopping, see for example Chapter 3 of \cite{peskir2006optimal}.

Once continuity of the value function is verified (Theorem \ref{thmContinuity_v}), the optimal stopping time is shown to exist (Corollary \ref{corrOptStop1}), and the martingale property of the Snell envelope of the discounted reward process is proven (see Theorem \ref{thmElkaroui81} \ref{enumElKarouiMartingale}, the connection between boundary value and optimal stopping problems can be established assuming enough regularity of the coefficients in \eqref{eqEDSFs}. A sufficient condition for the results of \cite{friedman1964partial} to hold is that the drift term in \eqref{eqEDSFs} is locally Hölder continuous, which is why Assumption \ref{assumpCHolderContinuous} is needed. The proof makes use of standard partial differential equation (PDE) results in solving the Dirichlet (or the first initial-boundary value) problem, see for instance \cite{friedman1964partial}, Theorem 3.4.9. 

The main difficulty in establishing the free boundary value problem resides in showing that the value function is continuous (Theorem \ref{thmContinuity_v}), which implies that the continuation region is open. Then, the Dirichlet problem (or the first initial-boundary value problem) can be posed in an open subset (typically a ball or a rectangle) of the continuation region. Since the boundary of the defined open subsets is sufficiently regular, standard PDE results guarantee that the Dirichlet problem admits a unique solution. Finally, probabilistic arguments are used to identify this solution as the value function.

For details of the proof of a similar result in the context of the American put in the Black-Scholes setting, see \cite{jacka1991optimal}, Proposition 2.6 or \cite{karatzas1998methods}, Theorem 2.7.7, among others. 
For more general results, see \cite{peskir2006optimal}, Sections 3.7.1 and 4.8.2 or \cite{jacka1992finite}, Proposition 3.1.

\begin{remark}
	Theorem \ref{thmFBproblemGen} only requires Assumption \ref{assumpCHolderContinuous} to hold. 
	In particular, even when the fee and surrender charge functions depend on time and the value of the account, it may still be possible to express the value function as a solution to a problem of the form \eqref{eqFBprobGen}.
\end{remark}

For the next remark, we recall that for an open set $O \subset \reals^d$, $d \in \mathbb N$, and for $1 \leq p < \infty$, $k \in \mathbb N$, the Sobolev space $\mathcal W^{k,p}(O)$ contains the $L^p(O)$ functions whose first $k$ weak partial derivatives have finite $L^p$-norm (see \cite{evans2010partial}, Chapter 5 for more details).
We further write that a function $f \in \mathcal W_{\text{loc}}^{k,p}(O)$ if $f \in \mathcal W^{k,p}(O_1)$ for every bounded domain $O_1$ with $\bar O_1 \subseteq O$.

\begin{remark}\label{remSobolev}
	It is possible to show that $v$ is the unique solution to \eqref{eqFBprobGen} among the functions $f \in \mathcal W_{loc}^{2,1}((0,T)\times \reals_+)$ satisfying $f(t,x) \leq G+x$ for all $(t,x) \in [0,T]\times\reals_+$, see for example Section 8.2 and Chapter 11 of \cite{pascucci2011pde}. 
	Indeed, under Assumption \ref{assumpCHolderContinuous}, all the coefficients of $\mathcal L_t$ are H\"{o}lder continuous, so that Hypotheses 8.1 and 8.3 of \cite{pascucci2011pde} are satisfied and their Theorem 9.48 holds.  
	In particular, this means that the weak derivatives $v_t$, $v_x$ and $v_{xx}$ are well-defined and integrable on $(0,T)\times \reals_+$.
\end{remark}

Assumption \ref{assumpCHolderContinuous} does not give much information on the shape of the surrender region. 
However, we show in Section \ref{sectSurrRegion} (Theorem \ref{thmMainSurrenderRegion}) that under further conditions on the fee and the surrender charge functions, the surrender region can be expressed in terms of the optimal surrender boundary $b: [0,T] \mapsto \reals_+$. 
Then, the free boundary value problem can be stated more explicitly.	

\begin{corollary}\label{thmFBproblem}
	If $L(t)<0$ for all $t\in[0,T)$, the value function $v$ solves the boundary value problem
	\begin{equation}
		\left\lbrace
		\begin{aligned}
			\mathcal{L}_tf(t,x)+f_t (t,x)-rf(t,x) & =0, &\quad x< b(t),\, t\in[0,\,T),\\
			f(t,x) & >\varphi(t,x),  &\quad x< b(t),\, t\in[0,\,T),\\
			f(t,\,x)& =\varphi(t,x), & x\geq b(t),\, t\in[0,\,T),\\
			f(T,x) & =\max(G,x), & x\in\reals_+,\\
		\end{aligned}\label{eqFBprob2}\right.
	\end{equation}
	with $b(t):=\inf\{x \in \reals_+ : (t,x) \in \Sr\}$. 
	In particular, the function $v$ is $C^{1,2}$ on $\mathcal C$.
\end{corollary}

Assuming $L(t)<0$ for all $t\in[0,T)$ in Corollary \ref{thmFBproblem} ensures that $\Sr\neq \emptyset$ (see Proposition \ref{propNonEmptyS_t}). This corollary is a direct consequence of Theorem \ref{thmFBproblemGen} and Theorem \ref{thmMainSurrenderRegion}\ref{thmMainSurrenderRegion1} of Section \ref{sectSurrRegion}.

\begin{remark}
	For each $t\in[0,T)$, $L(t)<0$ is equivalent to $\frac{\diff \ln g(t)}{\diff t}<c(t)$.
	From Gronwall's inequality,  this condition is satisfied if $g(t) < g(0) e^{\int_0^T c(s) \, \diff s}$ for all $t \in [0,T]$.
\end{remark}

\begin{remark}[Regularity of the value function]\label{rmkRegularity_v}
	We know from Theorem \ref{thmFBproblemGen} that $v \in C^{1,2}$ on $\Cr$, and from the definition of the value function on the surrender region that $v \in C^{1,2}$ on $\text{int}(\Sr)$. 
	It is known from the literature on optimal stopping that continuity of the second spatial derivative usually fails at the boundary $\partial \Cr$. 
	We also have from Remark \ref{remSobolev} that $v \in \mathcal W^{2,1}_{\text{loc}}((0,T)\times \reals_+)$.
	Then, by the Sobolev-Morrey embedding theorem (see for example \cite{pascucci2011pde}, Theorem 8.16), $v$ and $v_x$ are Hölder-continuous of exponent $\alpha \in (0,1)$ on  $[0,T] \times \reals_+$.
\end{remark}

Next, we show that $v$ solves a variational inequality. 

\begin{proposition}\label{propVI2}
	The value function $v$ defined in \eqref{eqAmOptVA} is a solution to the variational inequality
	\begin{equation}
		\max\left\lbrace\mathcal{L}_tv+v_t -rv,\varphi- v\right\rbrace= 0,\label{eqVI}
	\end{equation}
	with terminal condition $v(T,x)=\max(G,x)$ and all derivatives in the weak sense.
\end{proposition}

\begin{proof}
	Since $v \in \mathcal W_{\text{loc}}^{2,1}((0,T)\times\reals_+)$ (Remark \ref{remSobolev}), the (weak) partial derivatives in \eqref{eqVI} are well-defined.
	The Snell envelope $J = \{J_t\}_{0\leq t\leq T}$, with $J_t:=e^{-rt}v(t,\,F_t)$ of the discounted reward process being a supermartingale (Theorem \ref{thmElkaroui81}\ref{enumElKarouiSM}), 
	\begin{equation}
		\mathcal{L}_tv+v_t -rv\leq 0\label{eqVIinegality},
	\end{equation}
	with equality on $\Cr$ (Theorem \ref{thmFBproblemGen}). 
	The result follows since $v(t,x) \geq \varphi(t,x)$ (Lemma \ref{lemmaVprop1}) with equality on $\Sr$.
\end{proof}

\begin{remark}\label{remVI2}
	Proposition \ref{propVI2} holds as long as Assumption \ref{assumpCHolderContinuous} is satisfied.
\end{remark}

\begin{corollary}\label{corrBoundedWeakDerivatives}
	The partial derivatives (in the sense of distribution) $v_t$, $v_x$ and $v_{xx}$ are locally bounded, and $v_x$ is continuous on $[0,T)\times\reals_+$. 
\end{corollary} 

\begin{proof}
	From Proposition \ref{propLipchtitzV}, we have that the first order (weak) derivatives of the value function in $t$ and in $x$ are locally bounded. Using the convexity of $x\mapsto v(t,x)$ (Lemma \ref{lemmaVprop2}), the local boundedness of the first order derivatives and \eqref{eqVIinegality}, we can show that the second order derivative in $x$ is also locally bounded (see for instance \cite{lamberton1998}, Theorem 10.3.8 or \cite{jaillet1990variational}, Theorem 3.6 for details). 
	This last result is also known as the smooth fit condition.
	
	The continuity of $v_x$ stems from Theorem \ref{thmFBproblemGen} and the Sobolev-Morrey embedding theorem, as explained in Remark \ref{rmkRegularity_v}.
\end{proof}

\subsection{Surrender and Continuation Premium Representation }\label{sectionSurOptV1}

In this section and throughout the rest of the paper, derivatives are always considered in the weak sense.

Below, we derive two representations for the value function. The first one is akin to the early exercise premium representation (or integral representation) in the American option terminology. 

To the authors' knowledge, the second representation is new to the literature on variable annuities and American options. The particular form of the reward function involved in variable annuity pricing allows the decomposition of the value function in two terms: the current value of the reward process and an integral term which increases only when the account process is in the continuation region. 
Therefore, we call this second term the continuation premium.

\begin{proposition}\label{propSurPremRepresentation}
	The value function \eqref{eqAmOptVA} can be written as
	\begin{equation}\label{eqSurrPrem}
		v(t,x)=h(t,x)+ e(t,x),
	\end{equation}
	where $h:[0,T]\times \reals_+\rightarrow \reals_+$ is the present value of the maturity benefit, as defined in \eqref{eqVAnosurrender}, and $e:[0,T]\times \reals_+\rightarrow \reals_+ \cup \{0\}$ denotes the early surrender premium
	\begin{equation}\label{eqSurrBound}
		e(t,x):=\int_t^T \Big(c(s)g(s)-g_t(s)\Big) \e\Big[e^{-r(s-t)}F_s^{t,x}\ind_{\lbrace (s, F_s^{t,x})\in\mathcal{S}\rbrace}\Big]\diff s.
	\end{equation} 
\end{proposition}	
\begin{proof}
	Theorem \ref{thmFBproblemGen} ensures that $v$ is smooth enough to apply Itô's formula for generalized derivatives (see \cite{krylov2008controlled}, Theorem 2.10.1) to $\left\{e^{-rs}v(t,F_s^{t,x})\right\}_{t\leq s\leq T}$, yielding 
	\begin{align*}
		\diff (e^{-rs}v(s,F_s^{t,x}))&= e^{-rs}\left(\mathcal{L}_sv(s,F_s^{t,x})+v_t(s,F_s^{t,x})-r v(s,F_s^{t,x})\right)\diff s\\
		&  \quad + e^{-rs}\sigma F_s^{t,x} v_x(s, F_s^{t,x})\diff W_s
	\end{align*}
	for $0 \leq s \leq T$.
	Integrating from $t$ to $T$ on both sides and multiplying by $e^{rt}$, we get
	\begin{equation}
		\begin{split}\label{eqSurrPrem1}
			&e^{-r(T-t)}v(T,F_T^{t,x})\\
			&\quad=v(t,F_t^{t,x})\\
			&\qquad+ \int_t^T e^{-r(s-t)}\left(\mathcal{L}_sv(s,F_s^{t,x})+v_t (s,F_s^{t,x})-r v(s,F_s^{t,x})\right)\diff s\\
			&\qquad +\int_t^T e^{-r(s-t)}\sigma F_s^{t,x} v_x (s,\, F_s^{t,x})\diff W_s.
		\end{split}
	\end{equation}
	$\int_t^T e^{-r(s-t)}\sigma F_s^{t,x} v_x (s,\, F_s^{t,x})\diff W_s$ is a martingale since $|v_x(t,x)|\leq 1$ for all $(t,x)\in[0,T)\times\reals_+$, as per Proposition \ref{propLipchtitzV} and Corollary \ref{corrBoundedWeakDerivatives}, and ${\e[\int_0^T F_s^2\diff s]<\infty}$.
	Thus, taking the expectation on both sides of \eqref{eqSurrPrem1} yields
	\begin{equation*}
		\begin{split}
			v(t,x)&=\e\left[e^{-r(T-t)}v(T,F_T^{t,x})\right]\\
			&\qquad -\int_t^T \e\left[e^{-r(s-t)}\left(\mathcal{L}_sv(s,F_s^{t,x})+v_t(s,\,F_s^{t,x})-r v(s,\,F_s^{t,x})\right)\right]\diff s.
		\end{split}
	\end{equation*}
	Now, ${\e\left[e^{-r(T-t)}v(T,F_T^{t,x})\right]=\e\left[e^{-r(T-t)}\max(G,F_T^{t,x})\right]=h(t,\,x)}$, and by Corollary \ref{thmFBproblem}, $\mathcal{L}_tv(s,x)+v_t (s,x)-rv(s,x) =0$ for all $(s,\,x)\in\Cr$. Furthermore, for $s<T$, $v(s,\,x)=g(s)x$ when $(s,x)\in\Sr$. It follows that 
	\begin{equation*}
		\mathcal{L}v(s,F_s^{t,x})+v_t(s,\,F_s^{t,x})-r v(s,\,F_s^{t,x})=F_s^{t,x}\left[g_t(s)-c(s)g(s)\right]\ind_{\lbrace (s, F_s^{t,x})\in\mathcal{S}\rbrace},\quad \textrm{a.s.,}
	\end{equation*}
	which concludes the proof.
	\end{proof}

\begin{remark}\label{rmkSurrPremRep}
	If $L(t)<0$ for all $t\in[0,T)$, the surrender region has the particular shape $\Sr=\{(t,x)\in [0,T)\times \reals_+| x\geq b(t)\}$ for some $b(t)\geq Ge^{-r(T-t)}$. This will be shown in the next section in Theorem \ref{thmMainSurrenderRegion}. Under this assumption, the early surrender premium becomes
	\begin{equation}\label{eqSurrBound2}
		e(t,x):=\int_t^T \Big(c(s)g(s)-g_t(s)\Big) \e\Big[e^{-r(s-t)}F_s^{t,x}\ind_{\lbrace  F_s^{t,x}\geq b(s)\rbrace}\Big]\diff s.
	\end{equation}   
\end{remark}

Proposition \ref{propSurPremRepresentation} is a generalization of Theorem 1  and Equation (11) of \cite{bernard2014optimal} to time-dependent fee and surrender charge functions.

\begin{proposition}
	\label{thmContPrem}
	The value function \eqref{eqAmOptVA} can be written as
	\begin{equation}\label{eqLossRep}
		v(t,x)=xg(t,x)+f(t,x),
	\end{equation}
	where ${f:[0,T]\times\reals_+\rightarrow\reals_+ \cup \{0\}}$ is the continuation premium given by
	\begin{equation}\label{eqContPrem22}
		\begin{split}
			f(t,x)& = \e[e^{-r(T-t)}(G-F_T^{t,x})_+]\\
			&\qquad +\int_t^T \Big(g_t(s)-c(s)g(s)\Big) \e\Big[e^{-r(s-t)}F_s^{t,x}\ind_{\lbrace (s, F_s^{t,x})\in\mathcal{C}\rbrace}\Big]\diff s.
		\end{split}
	\end{equation} 
\end{proposition}

The proof of Proposition \ref{thmContPrem} is presented at the end of the section; it relies on the results below.

\begin{remark}\label{rmkContPremRep}
	\begin{sloppypar}
		If $L(t)<0$ for all $t\in[0,T)$, the continuation region is given by ${\Cr=\{(t,x)\in [0,T)\times \reals_+| x< b(t)\}}$ for some $b(t)\geq Ge^{-r(T-t)}$, see Theorem \ref{thmMainSurrenderRegion}. Under this assumption, the continuation premium in \eqref{eqContPrem22} becomes
		\begin{equation}\label{eqLossRepV3}
			\begin{split}
				f(t,x)&:=\e\left[e^{-r(T-t)}\left(G-F_T^{t,x}\right)_+\right] \\
				&\qquad+\int_t^T \Big(g_t(s)-c(s)g(s)\Big) \e\Big[e^{-r(s-t)}F_s^{t,x}\ind_{\lbrace  F_s^{t,x}<b(s)\rbrace}\Big]\diff s.
			\end{split}
		\end{equation}  
	\end{sloppypar}
\end{remark}

Proposition \ref{thmContPrem} decomposes the value function in terms of the immediate surrender value plus a term representing the value of holding on to the contract. This term, which we coin the \textit{continuation} premium, is the sum of the financial guarantee at maturity when the contract is held until $T$ and an integral term equal to the value added by keeping the contract until it is optimal to surrender. 
The representation in \eqref{eqLossRepV3} can be particularly helpful when developing numerical methods for approximating the value of a variable annuity contract, since it uses additional information on the shape of the continuation region. 

It is also possible to decompose the value function into the surrender benefit and the continuation premium in a more general setting, where Assumptions \ref{assumpCHolderContinuous} and \ref{assumpCandg} do not need to hold.
However, in this case, the continuation premium is written in terms of $\tau_t^x$, which complexifies computations. This result is presented in Lemma \ref{propContinuousPrem} below and is used in the proof of Proposition \ref{thmContPrem}.

\begin{lemma}\label{propContinuousPrem}
	The continuation premium in Proposition \ref{thmContPrem} can be written as
	\begin{equation}\label{eqContPrem}
		\begin{split}
			f(t,x)&:= \e\left[e^{-r(T-t)}(G-F_T^{t,x})_+\ind_{\{\tau_t^x= T\}}\right] \\
			&\qquad+\int_t^{T} \e\Big[ e^{-r(u-t)}F_u^{t,x} L(u, F_u^{t,x})\ind_{\lbrace u\leq \tau_t^x\rbrace}\Big]\diff u,
		\end{split}
	\end{equation}
	with $\tau_t^x$ defined in \eqref{eqOtpStop1}.
\end{lemma}
\begin{proof}
	The result is trivial for $t=T$. For the rest of the proof, fix $(t,x)\in[0,T)\times\reals_+$. Recall that $\tau_t^{x}$ is an optimal stopping time for $v(t,x)$ (see Corollary \ref{corrOptStop1}). Hence, we have   
	\begin{equation*}
		v(t,x)=\e\left[e^{-r(\tau_t^{x}-t)} \varphi(\tau_t^x, F_{\tau_t^x}^{t,x})\right].
	\end{equation*}
	\begin{sloppypar}
		Notice that the discounted reward process can be decomposed as 
		\begin{equation*}
			{e^{-r(s-t)}\varphi(s, F_{s}^{t,x})=Y_s^{t,x}+e^{-r(T-t)}(G-F_T^{t,x})_+\ind_{\{s=T\}}}
		\end{equation*}
		with ${Y^{t,x}=\{Y_s^{t,x}\}_{t\leq s\leq T}}$, with $Y_s^{t,x}:=e^{-r(s-t)}g(s,F_s^{t,x})F_s^{t,x}$ and observe that ${Y_T^{t,x}=e^{-r(T-t)}F_T^{t,x}}$ since ${g(T,x)=1}$ for all $x\in\reals_+$. The function $xg(t,x)$ is $C^{1,2}$, so we can apply Itô's formula to $Y$, which yields
	\end{sloppypar}
	\begin{align}
		&e^{-r(T-t)}\varphi(\tau_t^x, F_{\tau_t^x}^{t,x})\\
		&\quad= Y_{\tau_t^x}^{t,x}+ e^{-r(\tau_t^x-t)}(G-F_{\tau_t^x}^{t,x})_+\ind_{\{\tau_t^x=T\}}\nonumber\\
		\begin{split}
			&\quad= xg(t,x)+ e^{-r(T-t)}(G-F_T^{t,x})_+\ind_{\{\tau_t^x=T\}}\\
			&  \qquad +\int_t^{\tau_t^x} e^{-r(s-t)}F_s^{t,x}L(s,F_s^{t,x})  \diff s\\
			&  \qquad+ \int_t^{\tau_t^x} e^{-r(s-t)} \sigma F_s^{t,x}\left(g(s,F_s^{t,x})-g_x(s,F_s^{t,x})F_s^{t,x}\right) \diff W_s.\label{eqZtau1}
		\end{split}
	\end{align} 
	The final result is obtained by taking the expectation on both sides and using the zero-mean property of the stochastic integral and Doob's optional sampling theorem. To complete the proof, note that $f(t,x)\geq 0$ for all $(t,x)\in[0,T]\times\reals_+$, since $v(t,x)\geq\varphi(t,x)$ as per Lemma \ref{lemmaVprop1}.
\end{proof}

\begin{remark}
	As stated above, and as can be seen from its proof, Lemma \ref{propContinuousPrem} holds even when Assumptions \ref{assumpCHolderContinuous} and \ref{assumpCandg} are not satisfied.
\end{remark}

The proof of Proposition \ref{thmContPrem} also relies on the following Lemma, which is used to remove the dependence of the continuation premium on $\tau_t^x$.

\begin{lemma}\label{lemmaContPrem}
	For all $(t,x)\in[0,T]\times\reals_+$, 
	\begin{equation}\label{eqLossFctEquality}
		\begin{split}
			&\e\left[e^{-r(T-t)}(G-F_T^{t,x})_+\ind_{\{\tau_t^x\neq T\}}\right] \\
			&\qquad +\e\left[\int_{\tau_t^x}^T \Big(g_t(s)-c(s)g(s)\Big) e^{-r(s-t)} F_s^{t,x}\ind_{\lbrace (s, F_s^{t,x})\in\mathcal{C}\rbrace}\diff s\right]=0,
		\end{split}
	\end{equation}
	with $\tau_t^x$ defined in \eqref{eqOtpStop1}.
\end{lemma}
\begin{proof}
	The proof is trivial for $t=T$. For the rest of the proof, fix $(t,x)\in[0,T)\times\reals_+$.
	First, note that 
	\begin{equation}
		\begin{split}
			v(t,x)&=\e\left[e^{-r(T-t)} \varphi(T, F_{T}^{t,x})\right] \\
			&\quad+\e\left[e^{-r(\tau_t^x-t)}\varphi(\tau_t^x, F_{\tau_t^x})-e^{-r(T-t)} \varphi(T, F_{T}^{t,x})\right].\label{eqJt2}
		\end{split}
	\end{equation}
	Now recall that the discounted reward process admits the following decomposition 
	\begin{equation*}
		{e^{-r(s-t)}\varphi(s, F_{s}^{t,x})=Y_s^{t,x}+e^{-r(T-t)}(G-F_T^{t,x})_+\ind_{\{s=T\}}},
	\end{equation*}
	with $Y_s^{t,x}=e^{-r(s-t)}g(s)F_s^{t,x}$ for $t \leq s \leq T$. Hence, applying Itô's lemma to $Y^{t,x}=\{Y^{t,x}_s\}_{t \leq s \leq T}$, 
	we obtain   
	\begin{align}
		&e^{-r(T-t)}\varphi(T, F_{T}^{t,x})\nonumber\\ 
		&\qquad=e^{-r(T-t)}Y_{T}^{t,x}+ e^{-r(T-t)}(G-F_T^{t,x})_+\nonumber\\
		&\qquad= e^{-r(\tau_t^x-t)}Y_{\tau_t^x}+e^{-r(T-t)}(G-F_T^{t,x})_+\ind_{\{\tau_t^x=T\}}\nonumber\\
		&  \quad\quad 
		+\int_{\tau_t^x}^T e^{-r(s-t)}F_s^{t,x} (g_t(s)-c(s)g(s))\diff s \nonumber \\
		&  \quad\quad + \int_{\tau_t^x}^T e^{-r(s-t)}g(s)F_s^{t,x}\sigma \diff W_s+ e^{-r(T-t)}(G-F_T^{t,x})_+\ind_{\{\tau_t^x\neq T\}} \nonumber\\
		\begin{split}
			&\qquad= e^{-r(\tau_t^x-t)}\varphi(\tau_t^x, F_{\tau_t^x}^{t,x}) 
			+\int_{\tau_t^x}^T e^{-r(s-t)}F_s^{t,x} (g_t(s)-c(s)g(s))\diff s\\
			&  \qquad + \int_{\tau_t^x}^T e^{-r(s-t)}g(s)F_s^{t,x}\sigma \diff W_s 
			+ e^{-r(T-t)}(G-F_T^{t,x})_+ \ind_{\{\tau_t^x\neq T\}}.\label{eqZtau2}
		\end{split}
	\end{align}
	Replacing \eqref{eqZtau2} in \eqref{eqJt2} yields
	\begin{align}
		v(t,x)&= \e\left[e^{-r(T-t)} \varphi(T, F_{T}^{t,x})\right]+\e\left[\int_{\tau_t^x}^T e^{-r(s-t)}F_s^{t,x} (c(s)g(s)-g_t(s))\diff s\right]\nonumber\\
		& \quad-\e\left[e^{-r(T-t)}(G-F_T^{t,x})_+ \ind_{\{\tau_t^x\neq T\}}\right]\nonumber\\ 
		\begin{split} \label{eqSurrPrem22}
			\phantom{v(t,x)}&= h(t,x) +\e\left[\int_{\tau_t^x}^T e^{-r(s-t)}F_s^{t,x} (c(s)g(s)-g_t(s))\ind_{\{(s,F_s^{t,x}\in\Sr\}}\diff s\right]\\
			& \quad+\e\left[\int_{\tau_t^x}^T e^{-r(s-t)}F_s^{t,x} (c(s)g(s)-g_t(s))\ind_{\{(s,F_s^{t,x}\in\Cr\}}\diff s\right]\\
			& \quad-\e\left[e^{-r(T-t)}(G-F_T^{t,x})_+ \ind_{\{\tau_t^x\neq T\}}\right] ,
		\end{split}
	\end{align}  
	where $h(t,x)$ is defined as in \eqref{eqVAnosurrender}.
	Now note that
	\begin{equation*}
		\int_{t}^{\tau_t^x} e^{-r(s-t)}F_s^{t,x} (c(s)g(s)-g_t(s))\ind_{\{(s, F_s^{t,x})\in\Sr\}} \diff s=0\quad\textrm{ a.s.,}
	\end{equation*}
	since for $s\in [t,\tau_t^x]$, $F_s^{t,x}\in \Cr_s$, by definition of $\tau_t^x$ being the first entry time of $F_s^{t,x}$ in $\Sr_s$ between $t$ and $T$. Thus, 
	\begin{equation*}
		\begin{split}
			&\int_{\tau_t^x}^{T} e^{-r(s-t)}F_s^{t,x} (c(s)g(s)-g_t(s))\ind_{\{(s, F^{t,x}_s)\in\Sr\}} \diff s \\
			&\qquad= \int_{t}^{T} e^{-r(s-t)}F_s^{t,x} (c(s)g(s)-g_t(s))\ind_{\{(s, F_s^{t,x})\in\Sr\}} \diff s,\quad\textrm{ a.s.,}
		\end{split}
	\end{equation*}
	so \eqref{eqSurrPrem22} becomes
	\begin{equation}
		\begin{split}
			v(t,x)&= h(t,x)+e(t,x)\\
			& \quad+\e\left[\int_{\tau_t^x}^T e^{-r(s-t)}F_s^{t,x} (c(s)g(s)-g_t(s))\ind_{\{(s,F_s^{t,x}\in\Cr\}}\diff s\right]\\
			& \quad-\e\left[e^{-r(T-t)}(G-F_T^{t,x})_+ \ind_{\{\tau_t^x\neq T\}}\right],\label{eqSurrPrem4}
		\end{split}
	\end{equation}
	where $e(t,x)$ denotes the surrender premium defined in \eqref{eqSurrBound}.
	The final result is obtained by comparing \eqref{eqSurrPrem} and \eqref{eqSurrPrem4}.
	\end{proof}

We note that Lemma \ref{lemmaContPrem} requires Assumptions \ref{assumpCHolderContinuous} and \ref{assumpCandg} to hold since the proof makes use of the surrender premium representation \eqref{eqSurrPrem} in Proposition \ref{propSurPremRepresentation}. 

Lemmas \ref{propContinuousPrem} and \ref{lemmaContPrem} can now be used to prove Proposition \ref{thmContPrem}.

\begin{proof}[Proof of Proposition \ref{thmContPrem}] 
	When $t=T$, the proof is trivial. For the rest of the proof, fix $(t,x)\in [0,T)$. Using the continuation premium representation in \eqref{eqContPrem} and recalling from the proof of Lemma \ref{lemmaContPrem} that $\int_{t}^{\tau_t^x} e^{-r(s-t)}F_s^{t,x} (c(s)g(s)-g_t(s))\ind_{\{(s, F_s^{t,x})\in\Sr\}} \diff s=0$ a.s., we find that
	\begin{align*}
		v(t,x) &= xg(t)+ \e\left[e^{-r(T-t)}(G-F_T^{t,x})_+\ind_{\{\tau_t^x= T\}}\right]+\e\Big[\int_t^{\tau_t^x}  e^{-r(s-t)}F_s^{t,x} L(s) \diff s\Big]\\
		&= xg(t)+ \e\left[e^{-r(T-t)}(G-F_T^{t,x})_+(1-\ind_{\{\tau_t^x\neq T\}})\right]\\
		& \qquad\qquad +\e\Big[\int_t^{\tau_t^x}  e^{-r(s-t)}F_s^{t,x} L(s) \left(\ind_{\{F_s^{t,x}\in \Cr_s\}})+\ind_{\{F_s^{t,x}\in\Sr_s\}}\right) \diff s\Big]\\
		& = xg(t)+ \e\left[e^{-r(T-t)}(G-F_T^{t,x})_+\right] -\e\left[e^{-r(T-t)}(G-F_T^{t,x})_+\ind_{\{\tau_t^x\neq T\}}\right]\\
		& \qquad\qquad + \e\Big[\int_t^{\tau_t^x}  e^{-r(s-t)}F_s^{t,x} L(s) \ind_{\{F_s^{t,x}\in \Cr_s\}}\diff s\Big]\\
		&= xg(t)+ \e\left[e^{-r(T-t)}(G-F_T^{t,x})_+\right]\\
		& \qquad \qquad + \int_t^T L(s) \e\Big[e^{-r(s-t)}F_s^{t,x} \ind_{\{F_s^{t,x}\in \Cr_s\}}\Big]\diff s,
	\end{align*}
	where the last equality follows from Lemma \ref{lemmaContPrem}. 
	\end{proof}


\subsection{Characterization of the Surrender Region}\label{sectSurrRegion}

In this section, we study the impact of the fee and surrender charge functions on the shape of the surrender (and continuation) region.  
The main results of this section characterize specific types of surrender regions and are summarized in Theorem \ref{thmMainSurrenderRegion} below.
The rest of the section contains auxiliary results that are used in the proof of Theorem \ref{thmMainSurrenderRegion} and that are of interest on their own.
The proof of Theorem \ref{thmMainSurrenderRegion} is presented at the end of the section.

Henceforth, the $t$-section of the surrender (resp. continuation) region is denoted by $\Sr_t$ (resp. $\Cr_t$). That is, for $t \in [0,\,T)$,
\begin{equation*}
	\mathcal{S}_t=\lbrace x\in \reals_+:v(t,x)=\varphi(t,x)\rbrace
	\qquad \text{and} \qquad
	\mathcal{C}_t=\lbrace x \in \reals_+:v(t,x)>\varphi(t,x)\rbrace.
\end{equation*}

Although the shape of the surrender region can vary significantly (see for example Figures 4 and 5 of \cite{Mackay2017}, when the fee rate is constant and the (time-dependent) surrender charge function satisfies certain conditions, \cite{Mackay2014}, Appendix 2.A shows that the surrender region is a connected set. 
The first part of Theorem \ref{thmMainSurrenderRegion} below indicates that as soon as the surrender region is non-empty and the fee and surrender charge functions depend only on time, then $t$-section of $\Sr$ has the form $\Sr_t=[b(t),\infty)$ for some $b(t)\in\reals_+\cup\{\infty\}$.
This result  builds on \cite{jacka1991optimal}, Proposition 2.1 in the context of a bounded and time-homogeneous reward function in a time-homogeneous market model.

The second part of Theorem \ref{thmMainSurrenderRegion} confirms that the optimal stopping time obtained in Proposition \ref{propEDPSurreder} is unique.

\begin{theorem}\label{thmMainSurrenderRegion}
	Let $L$ be defined by \eqref{eqDefinitionLt}.
	\begin{enumerate}[label=(\roman*)]  
		\item\label{thmMainSurrenderRegion1}
		If $L(t)<0$ for all $t\in[0,T)$ then the $t$-sections $\mathcal{S}_t$ are of the form
		\begin{equation}
			\mathcal{S}_t = [b(t),\,\infty),\label{eqSurrBoundDefn}
		\end{equation}
		for some $b(t)\in \reals_+\cup\{\infty\}$ satisfying $b(t)\geq Ge^{-r(T-t)}$ and $t\in[0,T)$.
		\item\label{thmMainSurrenderRegion2}
		If $L(t)\geq 0$ for all $t\in[0,T)$, then $\Sr=\emptyset$ and $T$ is the unique optimal stopping time for \eqref{eqAmOptVA}.
	\end{enumerate}
\end{theorem}

The proof of Theorem \ref{thmMainSurrenderRegion} is presented at the end of the section, as it builds upon additional results presented below.

\begin{remark}\label{remarkMainThmNoAssumptions}
	The proof of part \ref{thmMainSurrenderRegion2} of Theorem \ref{thmMainSurrenderRegion} does not require Assumptions \ref{assumpCHolderContinuous} and \ref{assumpCandg} to hold. In this more general setting, if $L(t,x)\geq 0$ for all $(t,x)\in[0,T)\times\reals_+$, then $\Sr=\emptyset$ and $T$ is the unique optimal stopping time for \eqref{eqAmOptVA}.
\end{remark}

Under the assumptions of Theorem \ref{thmMainSurrenderRegion}, $b(t)$ is the smallest account value for which it is optimal to surrender the contract at time $t<T$, and for any account value greater than $b(t)$, it is also optimal to surrender, so that
\begin{equation}
	b(t) = \inf\left\{x\in\reals_+\Big|x\in \mathcal{S}_{t}\right\}=\inf\{\mathcal{S}_{t}\},
\end{equation}
with $b(t)=\infty$ if $\Sr_t=\emptyset$.
Under these assumptions, the continuation and the surrender regions can be expressed as 
\begin{equation*}
	\mathcal{C}=\left\lbrace (t,x) \in [0,T)\times\reals_+ \Big|x< b(t)\right\rbrace, 
\end{equation*}
and
\begin{equation*}
	\mathcal{S}=\left\lbrace (t,x) \in [0,T)\times\reals_+\Big|x\geq b(t)\right\rbrace,
\end{equation*}
respectively.
When $b(t)<\infty$ for all $t\in[0,T)$, the surrender boundary splits $[0,T)\times\reals_+$ in two regions: the surrender region is at or above the boundary, and the continuation region is below. 
It follows that the set $\mathcal{S}$ is connected. 
Henceforth, we say that the surrender region is of ``threshold type'' if for any $t\in[0,T)$, there exists a $b(t)<\infty$ such that $\Sr_t=[b(t),\infty)$.
Such a geometry for the surrender region can be explained, as in \cite{milevsky2001real}, by the fact that when the account value is low, it is financially advantageous for the policyholder to hold on to the contract since there is a significant chance that the guarantee will be triggered at maturity.
Since this guarantee is financed by the policyholder via the continuous fee, which reduces the net return on the account, there is a point above which it is no longer profitable to hold the contract and continue paying the fee; this threshold is the optimal surrender boundary.

The rest of this section presents further results on the (non-)emptiness of the surrender region and ends with the proof of Theorem \ref{thmMainSurrenderRegion}. Many of the results presented here are inspired by the work of \cite{villeneuve1999exercise}, which we adapted to the time-dependent and discontinuous reward function considered in \eqref{eqAmOptVA}. As pointed out by \cite{villeneuve1999exercise}, Remark 2.1, adapting some of their results to a time-dependent payoff is not trivial. 

Proposition \ref{propNonEmptyS_t} characterizes the non-emptiness of the surrender region. The first part involves $\text{relint} (\Sr_t)$, the relative interior of a $t$-section $\Sr_t$, which is defined as its interior within the affine hull of $\Sr_t$. If it is not empty, the relative interior of $\Sr_t$ forms a vertical line in the plane. If $\text{relint}(\Sr_t) \in \text{int}(\Sr)$, it contains some finite $x$ for which $(t,x) \in \text{int}(\Sr)$.
The first part of Proposition \ref{propNonEmptyS_t} \ref{enumStempty3} can be seen as a local version of the second part.


\begin{proposition} \label{propNonEmptyS_t}
	Let $L$ be defined by \eqref{eqDefinitionL}. 
	\begin{enumerate}[label=(\roman*)]  
		\item \label{enumStnonEmpty1}
		Let $\Sr \neq\emptyset$. For each $t\in [0,T)$, if $\text{relint}(\Sr_t) \subset \text{int} (\Sr)$, then $L(t) \leq 0.$
		\item \label{enumStempty3}
		If $L(t)<0$ for all $t\in[0,T)$, then 
		\begin{itemize}
			\item[$\bullet$]
			$\Sr \cap ([t_1,t_2) \times \reals_+) \neq \emptyset$ for any $0 \leq t_1<t_2 < T$, and
			\item[$\bullet$]
			$\Sr\neq \emptyset$ and $\Cr\neq \emptyset$.
		\end{itemize}  
	\end{enumerate}
\end{proposition}

\begin{proof}
	
	\begin{enumerate}[label=(\roman*)]
		
		\item 
		From the definition of the value function on $\Sr$ and of the surrender charge function $g$, $v\in C^{1,2}(\text{int}(\Sr))$. 
		Observe that for $x\in \text{relint}(\Sr_t)$, $v(t,x)=xg(t)$, so that $xL(t)=\mathcal{L}_tv+v_t-rv$. 
		Recall from Proposition \ref{propVI2} that $\mathcal{L}_tv+v_t-rv\leq 0$ for all $(t,x) \in (0,T) \times \reals_+$. The result follows since $x > 0$.
		
		\item 
		To show that $\Sr \cap ([t_1,t_2) \times \reals_+) \neq \emptyset$ for any $0 \leq t_1<t_2 < T$, we proceed by contradiction. Fix $0 \leq t_1 < t_2 < T$. 
		Suppose that $L(s)< 0$ for all $s\in[0,T)$ and $\Sr \cap ([t_1,t_2) \times \reals_+) = \emptyset$ . 
		This implies $\Cr_t=\reals_+$ for all $t\in[t_1,t_2)$, so that $v(t,x)-\varphi(t,x)>0$ for all $(t,x)\in [t_1,t_2)\times\reals_+$. Fix $(t,x)\in [t_1,t_2)\times\reals_+$. 
		Using the continuation premium representation \eqref{eqContPrem22} with $\Cr_t=\reals_+$ for all $t\in[t_1,t_2)$, we have 
		\begin{equation*}
			\begin{split}
				v(t,x)-\varphi(t,x)
				&=\e\left[e^{-r(T-t)}\left(G-F_T^{t,x}\right)_+\right]\\
				&\quad+\int_t^TL(s)\e\left[e^{-r(s-t)}F_s^{t,x}\ind_{\{F_s^{t,x}\in\Cr_s\}}\right] \diff s > 0, 
			\end{split}
		\end{equation*}
		so that,
		\begin{align}
			\e\left[e^{-r(T-t)}\left(G-F_T^{t,x}\right)_+\right]&>-\int_t^TL(s)\e\left[e^{-r(s-t)}F_s^{t,x} \ind_{\{F_s^{t,x}\in\Cr_s\}}\right] \diff s\nonumber\\
			&\geq- x \sup_{t\leq r<t_2} L(r) \int_{t}^{t_2}e^{-\int_t^{s}c(u)\diff u}\diff s.\label{eqProofStlocalNonEmpty}
		\end{align}
		\begin{sloppypar}
			Note that the function on the right-hand side is positive, equal to 0 when $x=0$, strictly increasing in $x$ (since $L(t)<0$ by assumption) and approaches infinity as $x \rightarrow \infty$. 
			Then, since the left-hand side of \eqref{eqProofStlocalNonEmpty} is bounded by $G e^{-r(T-t)}$,
			there must exist ${y\in\reals_+}$ such that  for all ${x\geq y}$, ${\e\left[e^{-r(T-t)}\left(G-F_T^{t,x}\right)_+\right]\leq-x \sup_{t\leq s<t_2} L(s) \int_{t}^{t_2}e^{-\int_t^{s}c(u)\diff u}\diff s}$, contradicting the strict inequality in \eqref{eqProofStlocalNonEmpty} which must hold for all ${x\in\reals_+}$. 
			We conclude that or $\Sr \cap ([t_1,t_2) \times \reals_+) \neq \emptyset$.
		\end{sloppypar}
		
		The first part of the second statement, which states that $\Sr$ is not empty, is obtained by taking $t_1 = 0$ and letting $t_2 = \lim_{\epsilon \rightarrow 0^+} T-\epsilon$ in the first statement. 
		To show that the continuation region is not empty either,
		we use the continuation premium representation \eqref{eqContPrem22}, so that
		\begin{align}
			\begin{split}
				v(t,x)-xg(t)& = \e\left[e^{-r(T-t)}(G-F_T^{t,x})_+\right] \\
				&\qquad +\int_t^TL(s)\e\left[e^{-r(s-t)}F_s^{t,x}\ind_{\{F_s^{t,x}\in\Cr_s\}}\right]\diff s,\label{eqContPremRepProof}
			\end{split}
		\end{align}
		for all $(t,x) \in [0,T]\times\reals_+$. 
		We want to show that there exists some $(t_1,x_1)\in [0,T) \times \reals_+$ satisfying ${v(t_2,x_2)>\varphi(t_2,x_2)}$.
		We proceed by contradiction. 
		
		Suppose that $L(s)< 0$ for all $s\in[0,T)$ and $\Cr=\emptyset$. Fix $t \in [0,T)$. Since $\Cr=\emptyset$, \eqref{eqContPremRepProof} becomes
		\begin{equation*}
			v(t,x)-\varphi(t,x)=v(t,x)-xg(t)= \e\left[e^{-r(T-t)}(G-F_T^{t,x})_+\right].
		\end{equation*}
		
		Now note that $\Cr_t=\emptyset$ implies $v(t,x)=\varphi(t,x)$ for all $x\in\reals_+$. However, we know that for all $x\in\reals_+$, $\e\left[e^{-r(T-t)}(G - F_T^{t,x})_+\right]>0$, leading to a contradiction.
		Thus, we conclude that $\Cr_t \neq \emptyset$, and thus $\Cr \neq \emptyset$.
	\end{enumerate}
\end{proof}

Proposition \ref{propEmptyS_t} characterizes the emptiness of the surrender region.

\begin{proposition}\label{propEmptyS_t}
	Let $L$ be defined by \eqref{eqDefinitionLt}.
	\begin{enumerate}[label=(\roman*)]  
		\item \label{enumStEmptyGlobal}
		For each $t\in [0,T)$, if $L(t) > 0$, then $\Sr_t = \emptyset$ or $\Sr_t \subseteq \partial \Sr$. 
		\item \label{enumStEmptyLocal}
		For any $[t_1,t_2)\subset[0,T)$ such that $t_1<t_2$, if $L(t)>0$ for all $t\in[t_1,t_2)$, then $\Sr \cap ([t_1,t_2) \times \reals_+)=\emptyset$.
		\item \label{enumStEmpty4}
		If $\Sr =\emptyset$, then $g(t)\leq e^{-\int_t^Tc(s)\diff s}$ for all $t\in[0,T)$.
	\end{enumerate}
\end{proposition} 

\begin{remark}\label{remStEmpty_assumptions}
	Some of the assumptions of Proposition \ref{propEmptyS_t} can be relaxed. The proof of part \ref{enumStEmptyGlobal} only requires Assumption \ref{assumpCandg}. Furthermore, part \ref{enumStEmptyLocal} holds in a general setting, without Assumptions \ref{assumpCHolderContinuous} and \ref{assumpCandg}. In this general setting, it holds that for any $[t_1,t_2)\subset[0,T)$ such that $t_1<t_2$, if $L(t,x)>0$ for all $(t,x)\in[t_1,t_2)\times\reals_+$, then $\Sr \cap ([t_1,t_2) \times \reals_+)=\emptyset$.    
\end{remark}

\begin{proof} 
	\begin{enumerate}[label=(\roman*)] 
		\item
		This statement is obtained by observing the contrapositive of Proposition \ref{propNonEmptyS_t} \ref{enumStnonEmpty1}.
		
		\item
		Since $\Sr \cap ([t_1,t_2) \times \reals_+)=\emptyset$ implies ${\Cr \cap ([t_1,t_2) \times \reals_+)=[t_1,t_2)\times\reals_+}$, we need to show that $v(t,x)>\varphi(t,x)$ for all $(t,x)\in[t_1,t_2)\times\reals_+$.
		Fix $(t,x)\in[t_1,t_2)\times\reals_+$ and recall from the proof of Lemma \ref{propContinuousPrem} that for any $0<t<s\leq T$
		\begin{equation*}
			e^{-r(s-t)}\varphi(s,F_s^{t,x})=Y_s^{t,x}+ e^{-r(T-t)}(G-F_T^{t,x})_+\ind_{\{s=T\}},
		\end{equation*}
		where $Y^{t,x}=\{Y_s^{t,x}\}_{t\leq s\leq T}$, with $Y_s^{t,x}=e^{-r(s-t)}g(s, F_s^{t,x})F_s^{t,x}$, is the discounted surrender value process. Hence, applying Itô's formula to $Y^{t,x}$, we find that
		\begin{align*}
			\e\left[e^{-r(t_2-t)}\varphi(t_2,F_{t_2}^{t,x})\right] & = xg(t) +\int_{t}^{t_2}  \e[L(s,F_s^{t,x}) e^{-r(s-t)}F_s^{t,x}]\diff s \\
			& \qquad\qquad\qquad+\e\left[e^{-r(T-t)}(G-F_T^{t,x})_+\ind_{\{t_2=T\}}\right]\\
			& > xg(t)=\varphi(t,x),
		\end{align*}
		where the last inequality follows since $L(t)>0$ for all $t \in[t_1,t_2)$. It follows that
		\begin{align*}
			v(t,x) & \geq \e\left[e^{-r(t_2-t)}\varphi(t_2,F_{t_2}^{t,x})\right] > \varphi(t,x),
		\end{align*}
		for all $(t,x)\in[t_1,t_2)\times\reals_+$, which completes the proof.
		
		\item
		Since $\Sr=\emptyset$, $\Cr=[0,T)\times\reals_+$, $v(t,x)-\varphi(t,x)>0$ for all $(t,x)\in[0,T)\times\reals_+$, which implies
		\begin{equation}
			\e\left[e^{-r(T-t)}(G-F_T^{t,x})_+\right]> x\left(g(t)-e^{-\int_t^T c(u) \diff u}\right).\label{eqProof }
		\end{equation}
		As observed in the proof of Proposition \ref{propNonEmptyS_t} \ref{enumStempty3}, the expectation on the left-hand side is a continuous and monotonically decreasing function of $x$, is equal to $Ge^{-r(T-t)}$ when $x=0$ and decreases to $0$ as $x\rightarrow\infty$. 
		Therefore, in order for the inequality to hold for all $(t,x) \in [0,T)\times\reals_+$, it must be true that
		\begin{equation}
			\label{eqRHSnegative}
			g(t)-e^{-\int_t^T c(u) \diff u} \leq 0
		\end{equation}
		for all $t \in [0,T)$.   
	\end{enumerate}
\end{proof}

\begin{remark}
	The first statements of Propositions \ref{propNonEmptyS_t} and \ref{propEmptyS_t} can apply to surrender charge functions that depend on both time and the value of the underlying process. 
\end{remark}

\begin{remark}\label{rmkL(t)=0}
	The statements of Proposition \ref{propNonEmptyS_t} \ref{enumStempty3} and \ref{propEmptyS_t} \ref{enumStEmptyGlobal} do not cover the case where $L(t)=0$. 
	Heuristically, when $L(t)=0$ for all $0\leq t\leq T$, the discounted surrender value process $\{Y_t\}_{0\leq t\leq T}$, with $Y_t=e^{-rt}F_tg(t)$, is a martingale, so one might expect that all stopping times $\tau$ such that $0\leq \tau\leq T$ are optimal, implying that ${\Sr=[0,T)\times\reals_+}$, see for instance \cite{bjork2009}, Proposition 21.2.
	However, this is incorrect; due to the time discontinuity of the reward function in \eqref{eqAmOptVA}, the discounted reward process $\{Z_t\}_{0\leq t\leq T}$, with $Z_t=e^{-rt}\varphi(t,F_t)$, is a submartingale (see Lemma \ref{lemmaVarphiSG}). 
	Hence, the policyholder can always profit (on average) from holding on to the contract because of the guaranteed amount at maturity. 
	The optimal stopping time is then $T$ and the surrender region is empty. 
	This illustrates another major difference  
	between the optimal stopping problem in \eqref{eqAmOptVA} and the one involved in the pricing of standard American options.
	
	The specificity of the VA pricing problem discussed above is also illustrated in the continuation premium representation in \eqref{eqContPrem22}.
	Indeed, when ${L(t)=0}$ for all $t\in[0,T]$, then for any $x\in\reals_+$, ${f(t,x)=\e[e^{-r(T-t)}(G-F_T^{t,x})_+]>0}$. 
	That is, the continuation premium $f$ is equal to the expected present value of the financial guarantee. 
	Hence, $v(t,x)=\varphi(t,x)+f(t,x)>\varphi(t,x)$ for all $(t,x)\in [0,T)\times\reals_+$, which implies that $\Sr=\emptyset$. 
\end{remark}

Propositions \ref{propNonEmptyS_t} and \ref{propEmptyS_t} can now be used to prove Theorem \ref{thmMainSurrenderRegion}.

\begin{proof}[Theorem \ref{thmMainSurrenderRegion}]
	
	To show part \ref{thmMainSurrenderRegion1}, let $L(t)<0$ for all $t\in[0,T)$.
	It follows from Proposition \ref{propNonEmptyS_t} \ref{enumStempty3} that $\Sr\neq\emptyset$, and thus there exists $t \in [0,T)$ for which $\Sr_t\neq \emptyset$. 
	Fix $t\in[0,T)$. If $\Sr_t=\emptyset$ then $b(t)=\infty$ and the proof is complete. 
	If $\Sr_t\neq \emptyset$, define
	\begin{equation*}
		b(t)=\inf\{x\in\reals_+|v(t,x)=\varphi(t,x)\}=\inf\{\Sr_t\}.
	\end{equation*}
	Since $\Sr_t$ is non-empty, we know that $b(t)<\infty$. Moreover, since $v$ and $\varphi$ are continuous on $[0,T)$ (see Theorem  \ref{thmContinuity_v}), $\Sr$ is a closed set that is bounded below, so that $b(t)\in\Sr_t$ and satisfies ${v(t,b(t))=\varphi(t,b(t))}$. 
	Thus, for $t<T,$
	\begin{align*}
		b(t)\geq g(t)b(t)=v(t,b(t))\geq \e\left[e^{-r(T-t)}\max(G,F_T^{t,b(t)})\right]\geq G^{-r(T-t)},
	\end{align*}
	since $g$ takes values in $(0,1]$.

	Next, we show that $\mathcal{S}_t= [b(t),\,\infty)$. 
	Fix $t\in [0,\,T)$ and note that if $\mathcal{S}_t=[b(t),\infty)$, then $\mathcal{C}_t=(0,b(t))$. Thus, we need to prove that $x\in\Cr_t\Rightarrow y\in\Cr_t$ for any $y<x$, which we do next.
	
	Recall that $\tau_{t}^x=\inf\left\lbrace t \leq s\leq T\,\big|\, v(s, F_s^{t,x})= \varphi(s, F_s^{t, x})\right\rbrace$ is optimal for $v(t,x)$. Hence, we have that
	\begin{align}
		0&\geq v(t,y)-v(t,\,x)\quad\text{ (since $x\mapsto v(t,x)$ is non-decreasing, see Lemma \ref{lemmaVprop2})}\nonumber\\
		&\geq \e\left[e^{-r(\xtau-t)}\varphi(\xtau,F_{\xtau}^{t,y})-e^{-r(\xtau-t)}\varphi(\xtau,F_{\xtau}^{t,x})\right]\nonumber\\
		\begin{split} \label{eqMaxF1}  
			&=\e\left[e^{-r(\xtau-t)}g(\xtau)e^{(r-\sigma^2/2)(\xtau-t)-\int_{t}^{\xtau}c(s)\diff s+\sigma W_{\xtau-t}} (y-x)\ind_{\lbrace \xtau <T\rbrace}\right] \\
			&  \quad + \e\left[e^{-r(\xtau-t)}\left\lbrace \max(G,\,F_{\xtau}^{t,y})-\max(G,\,F_{\xtau}^{t,x})\right\rbrace\ind_{\lbrace \xtau =T\rbrace}\right]. 
		\end{split}
	\end{align} 
	For all $\omega\in\Omega$ and $\tau\in\mathcal{T}_{t,T}$, we have 
	$F_{\tau(\omega)}^{t,x}(\omega)> F_{\tau(\omega)}^{t,y}(\omega)$, since $y< x$.
	It follows that $(G-F_{\tau}^{t,y})_+-(G-F_{\tau}^{t,x})_+\geq 0$, for all $\omega\in\Omega$, so that
	\begin{align}
		&  \e\left[e^{-r(\xtau-t)}\left\lbrace \max(G,\,F_{\xtau}^{t,y})-\max(G,\,F_{\xtau}^{t,x})\right\rbrace\ind_{\lbrace \xtau =T\rbrace}\right]\nonumber\\ 
		&  \quad \geq \e\left[e^{-r(\xtau-t)} \left\lbrace F_{\xtau}^{t,y}-F_{\xtau}^{t,x} \right\rbrace\ind_{\lbrace \xtau =T\rbrace}\right]\nonumber\\
		& \quad = \e\left[e^{-r(\xtau-t)}g(\xtau)e^{(r-\sigma^2/2)(\xtau-t)-\int_{t}^{\xtau}c(s)\diff s+\sigma W_{\xtau-t}} (y-x)\ind_{\lbrace \xtau =T\rbrace}\right],\label{eqMaxF2}
	\end{align} 
	since $g(T)=1$. 
	Combining \eqref{eqMaxF1} and \eqref{eqMaxF2} yields
	\begin{align}
		0	\geq v(t,y)-v(t,\,x)&\geq (y-x)  \e\left[e^{-r(\xtau-t)}g(\xtau)e^{(r-\sigma^2/2)(\xtau-t)-\int_{t}^{\xtau}c(s)\diff s+\sigma W_{\xtau-t}}\right]\nonumber\\
		&= g(t)(y-x)  \e\left[\frac{g(\xtau)}{g(t)}e^{-\int_{t}^{\xtau}c(s)}e^{\sigma W_{\xtau-t}-\frac{\sigma^2}{2}(\xtau-t)\diff s} \right]\nonumber\\
		&= g(t)(y-x)  \e\left[e^{\int_t^{\xtau}\frac{\diff \ln g(s)}{\diff s}-c(s)\diff s}e^{\sigma W_{\xtau-t}-\frac{\sigma^2}{2}(\xtau-t)} \right]\nonumber\\
		&> g(t) (y-x) \e\left[e^{\sigma W_{\xtau-t}-\frac{\sigma^2}{2}(\xtau-t)}  \right] 
		\\
		&= g(t)(y-x).\nonumber
	\end{align}
	Finally, since $x\in\Cr_t$, we have $v(t,\,x)>g(t) x$, so that $v(t,y) > g(t)y$, which implies $y\in\Cr_t$. 
	
	For part \ref{thmMainSurrenderRegion2},  we have from Proposition \ref{propEDPSurreder} that $v(t,x)=h(t,x)$, for all $(t,x)\in[0,T]\times \reals_+$. Now using Itô's lemma on the discounted surrender value process and the zero-mean property of the stochastic integral, we find that for any $(t,x) \in [0,T)\times\reals_+$, 
	\begin{align*}
		v(t,x)&=\e\left[e^{-r(T-t)}\max(G,F_T^{t,x})\right]\\
		&=\e\left[e^{-r(T-t)} g(T,F_T^{t,x}) F_T^{t,x}\right]+\e\left[e^{-r(T-t)}(G-F_T^{t,x})_+\right]\\
		&=g(t,x)x+\e\left[\int_t^Te^{-r(s-t)}F_s^{t,x} L(s,F_s^{t,x}) \diff s\right]+\e\left[e^{-r(T-t)}(G-F_T^{t,x})_+\right]\\
		&>g(t,x)x.
	\end{align*}
	\begin{sloppypar}
		Hence, $\Sr=\{(t,x)\in[0,T)\times\reals_+| v(t,x)=\varphi(t,x)\}=\emptyset$, so that ${\tau_t^x=\inf\{t\leq s\leq T | v(s,F_s^{t,x})=\varphi(s,F_s^{t,x})\}=T}$. Now since $\tau_t^x$ is the smallest optimal stopping time for \eqref{eqAmOptVA},  as per Theorem \ref{thmElkaroui81} \ref{enumElKarouiMinTau}, we conclude that it is unique.
	\end{sloppypar}

\end{proof}

\subsection{Equivalence of the Optimal Stopping Problems}\label{sectionEquivalenceProblem}
In this section, we compare the optimal stopping problems with continuous and discontinuous reward functions introduced in Section \ref{sectOptStop}, and provide a simple condition under which the two problems lead to the exact same surrender regions and optimal stopping times.

\begin{proposition}\label{propSurrRegionsEqual}
	Let $L$ be defined by \eqref{eqDefinitionLt}.
	If $L(t)<0$ for all $t\in[0,T)$, then $\mathcal{S}=\tilde{\mathcal{S}}$.
\end{proposition}
The proof of Proposition \ref{propSurrRegionsEqual} relies on the following lemma.
\begin{lemma}\label{lemmaValueFonction}
	If $L(t) < 0$ for all $t\in[0,T)$,
	${v(t,x)>h(t,x)}$ for all $(t,x)\in [0,T)\times\reals_+$, where ${h(t,x)=\e\left[e^{-r(T-t)}\max(G,F_T^{t,x})\right]}$.
\end{lemma}
\begin{proof}
	We proceed by contradiction. Suppose that $L(t)<0$ for all $t\in[0,T)$ and that $v(t,x)=h(t,x)$ for some $(t,x)\in[0,T)\times \reals_+$. Using Itô's lemma on the discounted surrender value process and the zero-mean property of the stochastic integral, we find that
	\begin{align*}
		v(t,x)&=\e\left[e^{-r(T-t)}\max(G,F_T^{t,x})\right]\\
		&=\e\left[e^{-r(T-t)} g(T) F_T^{t,x}\right]+\e\left[e^{-r(T-t)}(G-F_T^{t,x})_+\right]\\
		&=xg(t)+\e\left[\int_t^Te^{-r(s-t)}F_s^{t,x} L(s) \diff s\right]+\e\left[e^{-r(T-t)}(G-F_T^{t,x})_+\right]\\
		&\leq xg(t)+x \sup_{t\leq s\leq T} L(s) \int_t^T e^{-\int_t^s c(u)\diff u}\diff s + \e\left[e^{-r(T-t)}(G-F_T^{t,x})_+\right]
	\end{align*}
	Now recall from Lemma \ref{lemmaVprop1} \ref{enumVprop2} that $v(t,x)\geq xg(t)$ for all $(t,x)\in [0,T)\times \reals_+$. Hence, it follows that
	\begin{equation*}
		xg(t)\leq v(t,x)\leq xg(t)+x \sup_{t\leq s\leq T} L(s) \int_t^T e^{-\int_t^s c(u)\diff u}\diff s + \e\left[e^{-r(T-t)}(G-F_T^{t,x})_+\right].
	\end{equation*}
	For the last inequality to be satisfied, it must be true that
	\begin{equation*}
		\e\left[e^{-r(T-t)}(G-F_T^{t,x})_+\right] \geq -x \sup_{t\leq s\leq T} L(s) \int_t^T e^{-\int_t^s c(u)\diff u}\diff s. 
	\end{equation*}
	Using arguments similar to those of the proof of Proposition \ref{propNonEmptyS_t}\ref{enumStempty3}, we find a contradiction and therefore conclude that $v(t,x) > h(t,x)$ for all $(t,x) \in [0,T)\times \reals_+$.
	\end{proof}
We can now prove Proposition \ref{propSurrRegionsEqual}
\begin{proof}[Proposition \ref{propSurrRegionsEqual}]
	To show that $\mathcal{S}\subseteq\tilde{\mathcal{S}}$, fix ${(t,x)\in\mathcal{S}}$, and observe that ${\varphi(t,x)=v(t,x)\geq h(t,x)}$. Hence, ${v(t,x)=\max(\varphi(t,x),h(t,x))}$, and thus $(t,x) \in \tilde{\mathcal S}$, by definition of $\tilde \Sr$
	
	To show that $\mathcal{S}\supseteq\tilde{\mathcal{S}}$, we show $\mathcal{C}\subseteq\tilde{\mathcal{C}}$. 
	Fix $(t,x)\in\Cr$, so that ${v(t,x)>\varphi(t,x)}$.
	Since $L(t)<0$ for all $t\in[0,T)$, we have by Lemma \ref{lemmaValueFonction} that ${v(t,x)>h(t,x)}$ for all $(t,x)\in [0,T)\times\reals_+$.
	Thus, $v(t,x)> \max(\varphi(t,x), h(t,x))$, which confirms that $(t,x)\in\tilde{\Cr}$.
\end{proof}
The next corollary establishes the equality of $\tau_t^x$ and $\tilde{\tau}_t^x$ under a simple condition, complementing the previous results of Lemma \ref{lemmaStoppingTimes}.
\begin{corollary}\label{corrEqualOptStop}
	If $L(t)<0$ for all $t\in[0,T)$, then $\tau_t^x=\tilde{\tau}_t^x$.
\end{corollary}
\begin{proof}
	The stopping times $\tau^{x}_t$ and $\tilde{\tau}^{x}_t$ can be written as
	\begin{equation*}
		\begin{array}{cc}
			\tau^{x}_t & =\inf\{t\leq s\leq T| (s,F_s^{t,x})\in\mathcal{S}\},\\
			\tilde{\tau}^{x}_t &= \inf\{t\leq s\leq T| (s,F_s^{t,x})\in\tilde{\mathcal{S}}\}.
			\end {array}
		\end{equation*}
		The proof then follows immediately from  Proposition \ref{propSurrRegionsEqual}.
		\end{proof}
	
	From the results above, we can conclude that when the fee and surrender charge functions only depend on time (Assumption \ref{assumpCandg}) and satisfy other regularity assumptions (Assumption \ref{assumpCHolderContinuous}), if $L(t)<0$ for all $t\in[0,T)$, the optimal stopping problem with the continuous reward function \eqref{eqVASurOpt3} is equivalent to the optimal stopping problem with the discontinuous reward function \eqref{eqAmOptVA}. 
	That is, the two problems lead to the same value function (Theorem \ref{thmContRep}), surrender region (Proposition \ref{propSurrRegionsEqual}), and optimal stopping time (Corollary \ref{corrEqualOptStop}). 
	It also follows that the optimal surrender boundary and the continuation region will be the same for the two problems.
	
	To provide some intuition for this result, we remark that the condition $L(t) < 0$ is associated with a non-empty surrender region (see Proposition \ref{propNonEmptyS_t}); at any point during the life of the contract, there is always a possibility that early surrender will become optimal in the future. 
	Then, the value of the VA contract is strictly above that of the EPV of the maturity benefit, and if surrender occurs it will be due to the contract value being equal to the surrender benefit. 
	When this is the case, all three stopping times coincide (see the discussion after the proof of Lemma \ref{lemmaStoppingTimes}).
	

\section{Numerical Examples}\label{sectNumResults}

\subsection{Time-dependent fee function}

In this section, we present two simple examples to highlight specific features that can be observed in the optimal stopping region associated with variable annuity pricing: its disconnectedness and the discontinuity of its boundary.

We consider the time-dependent fee functions $c_1$ and $c_2$, given by
\begin{align*}
	c_1(t) &=\frac{1}{100}\left(0.00889t^2-0.1330t + 1\right), \qquad
	c_2(t) &=\frac{1}{100}\left(-0.0333 t + 1\right),
\end{align*}
for $0 \leq t \leq T$, with $T=15$. The surrender charge function is set to $g(t)=e^{-0.0055 (T-t)}$, a commonly used form in the variable annuity literature.
The resulting surrender charge at inception of the contract is close to $8\%$, and it goes down to less than $0.5\%$ at $t=14$, one year before maturity.

The fee functions were chosen to illustrate some of the results obtained in Section 4.
Function $c_1$ shows the impact of a change in the sign of $L(t)$; it is negative on $[0,5.17]\cup[9.79,15]$ and positive elsewhere.
Function $c_2$ was chosen such that $L(t) > 0$ for $t\in(13.5,T]$ for $c_2$, thus eliminating the risk-neutral surrender incentive in the last 1.5 year of the contract.
In both cases, the annual fee rates vary between $0.5\%$ and $1.0\%$.
In practice, the fee is typically set as a constant.
Numerical examples in this setting are provided, for example, in \cite{bernard2014optimal}, \cite{Mackay2014} and \cite{Mackay2017}.

\begin{figure}[h!]
	\centering
	\begin{tabular}{cc}
		\includegraphics[scale=0.3]{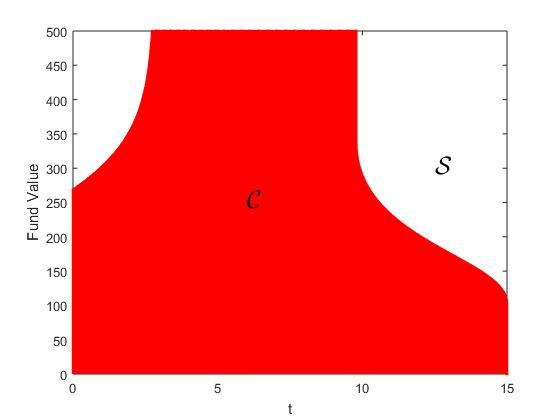}&
		\includegraphics[scale=0.3]{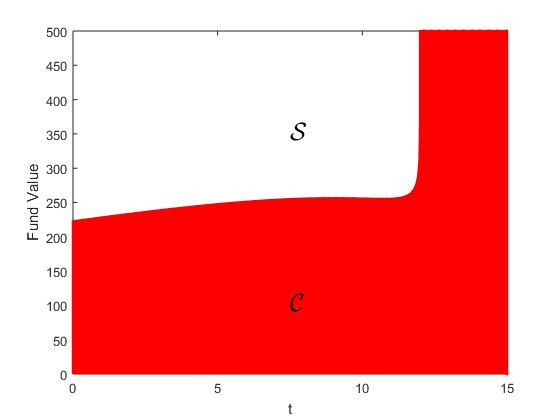}\\
		\textsc{(a)} $c_1$ & \textsc{(b)} $c_2$  \\
	\end{tabular}
	\caption{\small{The continuation region (in red) of the optimal stopping problem with the discontinuous reward function \eqref{eqAmOptVA}. The value of the variable annuity contract is approximated using the continuous-time Markov chain approximation described in \cite{MackayVachonCui2022VaCTMC}. 
			Market and VA parameters are $r=0.03$, $\sigma=0.2$, $F_0=G=100$ and $T=15$.}}\label{figComparisonContRegions}
\end{figure}

Figure \ref{figComparisonContRegions} shows the continuation regions (in red) associated with the optimal stopping problem \eqref{eqAmOptVA} for the two fee functions $c_1$ and $c_2$. 
The surrender regions (in white) are empty between years $5$ and $10$ for the fee function $c_1$, and between years $13.5$ and $15$ for $c_2$. 
This corresponds to the time intervals during which $L(t)>0$, illustrating a result of Proposition \ref{propEmptyS_t} \ref{enumStEmptyLocal}. Thus, the optimal surrender boundary is discontinuous when the fee is given by $c_1$. Other examples of the influence of the fee function on the surrender region exist, for example, in \cite{Mackay2017}, Figures 4 and 5. However, disconnected sets for the surrender region have not been observed in prior numerical work on similar problems, see among others \cite{bernard2014optimal}, \cite{Mackay2017}, \cite{kang2018optimal}, and \cite{MackayVachonCui2022VaCTMC}. 

Another important feature highlighted in the second example (with fee function $c_2$) is that $\lim_{t\rightarrow T} b(t)=G$ does not hold. 
This is because the surrender region is empty for $t>13.5$, just before maturity (since $L(t)>0$).
Moreover, in this particular case, since $\Sr\cap ([13.5,T)\times \reals_+)=\emptyset$, we have $v(t,x)=h(t,x)$ for $(t,x)\in [13.5,T)$, so that  $\tilde{\Sr}\cap ([13.5,T)\times \reals_+)=[13.5,T)\times \reals_+$, providing another example in which the optimal stopping problems with the discontinuous and the continuous reward functions have different surrender (and continuation) regions.

Finally, we note that in our figures, the surrender region becomes empty slightly before $t=5$ for the fee function $c_1$, and before $t=13.5$ for $c_2$.
This is likely due to the $y$-axis stopping at 500; it may be that there exists a surrender region at $t=5$ for account values that are (much) higher than 500, and that the surrender boundary is continuous in a general sense, ``reaching infinity'' continuously when the sign of $L$ changes.
Thus, we conjecture that $L(t) \leq 0$ if and only if $\Sr_t$ is non-empty, but this is not something we have been able to show.

An economic interpretation for the condition $L(t)>0$, which leads to the surrender region being empty, is that the surrender charge to be paid upon surrender is greater than the expected future fee payment.
Surrenders then become more expensive and the policyholder prefers to hold on to her contract.
This is also discussed in \cite{Mackay2017} in a less general setting.

\subsection{State-dependent fee function}

Next, to show the impact of a fee function that depends on the value of the underlying account, we define
\begin{align*}
	c_3(t,x) =  \frac{0.12 e^{150-x}}{1+e^{150-x}}.
\end{align*}
This function was chosen to resemble the discontinuous fee studied in \cite{Mackay2017} while satisfying Assumption \ref{assumpCHolderContinuous}. 
The resulting fee stays very close to 0.012 until the investment account $F_t$ reaches values around 150, at which point it drops quickly to 0.
The idea behind such a state-dependent fee is to reduce the surrender incentive by eliminating the fee when the value of the maturity guarantee is low.
Following \cite{Mackay2017}, we also use $g(t,x) =  1-0.05\left(1-\frac tT\right)^3$, $F_0 = G = 100$, $T=10$, $\sigma = 0.165$ and $r=0.03$.
The resulting optimal surrender region is presented in Figure \ref{figSDFee}.

\begin{figure}[h!]
	\centering
	\includegraphics[scale=0.3]{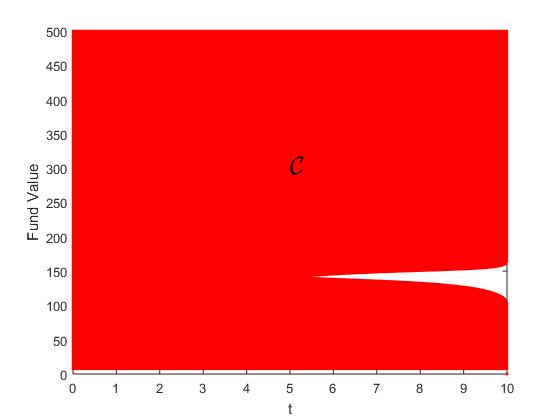}
	\caption{\small{Continuation region (in red) of the optimal stopping problem \eqref{eqAmOptVA}. The value of the variable annuity contract is approximated using the continuous-time Markov chain approximation described in \cite{MackayVachonCui2022VaCTMC}. 
			Market and VA parameters are $r=0.03$, $\sigma=0.165$, $F_0=G=100$ and $T=10$.}}\label{figSDFee}
\end{figure}

This specific state-dependent fee removes the optimal surrender incentive on $0 \leq t \leq 5$ (that is, $\Sr \cap [0,5] = \emptyset$). 
Indeed, whenever the $t$-section of the surrender region is non-empty, it has the form $(b_1(t),b_2(t))$, for $b_1(t)$, $b_2(t) < \infty$.
The surrender region is not of the so-called ``threshold'' type; if the account value increases quickly at the beginning of the contract and remains high until maturity, it may never be optimal to surrender.
However, it is optimal to do so for some values of $x$ on $t \in (5,10]$.
This is due to the vanishing fee as $F_t$ increases.

\section{Conclusion}\label{sectionConclu}
In this paper, we perform a rigorous theoretical analysis of the value function involved in the pricing of a variable annuity contract with guaranteed minimum maturity benefit under the Black-Scholes setting with general fee and surrender charge functions. Because of the time dependence, the discontinuity, and the unboundedness of the reward function, many of the standard results in optimal stopping theory do not apply directly to our problem. We show that the optimal stopping problem in \eqref{eqAmOptVA} admits another representation with a continuous reward function \eqref{eqVASurOpt3}, which facilitates the study of the regularity of the value function. 

In particular, the continuous reward representation allows us to directly apply results from optimal stopping theory to obtain the continuity of the value function.  We also prove its convexity in the state variable, its Lipschitz property in $x$ and characterize its continuity in $t$. Then, we derive an integral expression for the early surrender premium, generalizing the results of \cite{Bernard2014} to general time-dependent fee and surrender charge functions. This second representation of the value function is also known as the early exercise premium representation in the American option pricing literature. From there, we develop a third representation for the value function in terms of the current surrender value and an integral expression that only takes value in the continuation region. We call this third representation the continuation premium representation. 

The continuation premium representation turns out to be very helpful in studying the shape of the surrender region. We show that the (non-)emptiness of the surrender region depends on an explicit condition that is expressed solely in terms of the fee and the surrender charge functions. This result is new to the literature and provides a better understanding of the interaction between the fees and the surrender penalty on surrender incentives. When the surrender region is nonempty, we show that $t$-sections of $\Sr$ are of the form $\Sr_t=[b(t),\infty)$, for some $b(t)\in \reals_+ \cup \{\infty\}$.  Investigating the regularity of the optimal surrender boundary theoretically is left as future research.

\section*{Appendix: Proof of Proposition \ref{propLipchtitzV}}

\begin{proof}
	\begin{enumerate}[label=(\roman*)]
		\item 
		We first observe that under Assumption \ref{assumpCandg}, the function $x \mapsto \varphi(t,x)$ satisfies $\abs{\varphi(t,x)-\varphi(t,y)}\leq \abs{x-y}$ for all $t \in [0,T]$. 
		Suppose $x>y$.
		Since $x \mapsto \varphi(t,x)$ is non-decreasing (see Lemma \ref{lemmaVprop2}), we have
		\begin{align*}
			|v(t,x)-v(t,\,y)|&=v(t,x)-v(t,\,y) \\
			&\leq \e\left[e^{-r(\xtau-t)}\left|\varphi(\xtau, F_{\xtau}^{t,x})-\varphi(\xtau, F_{\xtau}^{t,y})\right|\right]\\
			&\leq \e\left[e^{-r(\xtau-t)}\left|F_{\xtau}^{t,x}- F_{\xtau}^{t,y}\right|\right] \\
			&\leq\e\left[e^{-r(\xtau-t)}\left|xe^{(r-\sigma^2/2)(\xtau-t)-\int_0^{\tau_t^x-t}c(t+s)\diff s+\sigma W_{\xtau-t}}\right.\right.\\
			& \quad\left.\left.- ye^{(r-\sigma^2/2)(\xtau-t)-\int_0^{\tau_t^x-t}c(t+s)\diff s+\sigma W_{\xtau-t}}\right|\right]\\
			&\leq |x-y|\e\left[ e^{\sigma W_{\xtau-t}-\sigma^2(\xtau-t)/2}\right]\\
			&=|x-y|.
		\end{align*}
		The last equality follows from an application of Doob's optional sampling theorem.
		
		\item 
		To show (ii), we need some auxiliary results. 
		
		First, if $\ttau$ is a stopping time taking values in $[0,1]$ and $c: \reals_+ \mapsto [0,1]$, then for $0\leq s\leq t \leq T$, it can be shown (by considering separately the cases $t \leq s+\ttau (T-s)$ and $s+\ttau(T-s) < t$) that
		\begin{equation}
			\Big|\int_s^{s+\ttau(T-s)} c(u)\diff u-\int_t^{t+\ttau(T-t)} c(u)\diff u\Big| \leq 3|t-s|.
			\label{eq:boundIntcu}
		\end{equation}
		
		
		Moreover, defining $M_1 := \sup_{0\leq s\leq 1} |W_s|$, we have that
		\begin{align}
			&\left|\sigma\sqrt{T-t} W_{\ttau(\omega)} (\omega)+(r-\sigma^2/2)\ttau(\omega)(T-t)-\int_0^{\ttau(\omega)(T-t)}c(t+u)\diff u\right|\nonumber\\
			& \qquad\leq \sigma \sqrt{T}M_1(\omega)+|r-\sigma^2/2|T+\int_0^{T}c(u)\diff u.\label{eqLipsW}
		\end{align}
		
		We also observe that if $|z|\leq r$, $|w|\leq r$, for some $r>0$, then 
		\begin{equation}
			|e^z-e^w|\leq e^r |z-w|.\label{eqLipsEXP}
		\end{equation}
		This is known as the Lipschitz property of the exponential function on a disk and follows from $z^{n+1}-w^{n+1}=(z-w)\sum_{k=0}^nz^k w^{n-k}$ for all $n \in \mathbb N$.

		Finally, note that for any $t \in [0,T)$, $v(t,x)$ can be written as
		\begin{align*}
			&\sup_{\ttau \in [0,1]} \e\left[ e^{-r\ttau(T-t)} 
			\varphi\left(t+\ttau(T-t), x e^{(r-\sigma^2/2)\ttau(T-t) + \int_t^{t+\ttau(T-t)} c(u)\, \diff u + \sigma (W_{t+\ttau(T-t)} - W_t)}\right)\right]\\
			&\quad= \sup_{\ttau \in [0,1]} \e\left[ e^{-r\ttau(T-t)} 
			\varphi\left(t+\ttau(T-t), x e^{(r-\sigma^2/2)\ttau(T-t) + \int_t^{\ttau(t+T-t)} c(u)\, \diff u + \sigma \sqrt{T-t} W_{\ttau}}\right)\right],
		\end{align*}
		since $(W_{t+\ttau(T-t)} - W_t)$ and $\sqrt{T-t} W_{\ttau}$ have the same distribution.
		
		Going forward, to simplify the notation, we define
		\begin{align*}
			\tilde F^{t,x}_{t+\ttau(T-t)} \coloneqq
			x e^{(r-\sigma^2/2)\ttau(T-t) + \int_0^{\ttau(T-t)} c(t+u)\, \diff u + \sigma \sqrt{T-t} W_{\ttau}},
		\end{align*}
		for $(t,x) \in [0,T]\times\reals_+$. 
		
		To prove (ii), we fix $x \in \reals_+$ and define
		$\ttau_t \coloneqq \frac{\tau_t^x-t}{T-t}$
		for any $t \in [0,T]$ and where $\xtau$ is defined as in Corollary \ref{corrOptStop1}. 
		We drop the superscript $x$ for ease of notation.
		Then, the value function $v(t,x)$ can be written as 
		\begin{equation*}
			v(t,x) = \e\left[e^{-r\ttau_t(T-t)} \varphi\left(t+\ttau_t(T-t), \tilde F^{t,x}_{t+\ttau(T-t)}\right)\right].
		\end{equation*}
		
		Fix $0\leq s \leq t \leq T$.
		Since $v$ is non-monotone in $t$ (see Remark \ref{rmkNonMonotonicity_t}), we consider two cases: $v(s,x)\geq v(t,x)$ and $v(s,x) < v(t,x)$ .
		
		Case 1: $v(s,x)\geq v(t,x)$. 
		Recall from Corollary \ref{corrOptStop1} that $\xtaus$ is optimal for $v(s,x)$.
		Then, we can write
		\begin{align*}
			&|v(s,x)-v(t,\,x)| =v(s,x)-v(t,\,x) \\
			&  \quad \leq \e\left[e^{-r\ttau_s(T-t)}\left(\varphi\left(s+\ttau_s(T-s), \tilde F_{s+\ttau_s(T-s)}^{s,x}\right)-\varphi\left(t+\ttau_s(T-t), \tilde F_{t+\ttau_s(T-t)}^{t,x}\right)\right)\right]\\
			&  \quad \leq\e\left[\left(g(s+\ttau_s(T-s)) \tilde F_{s+\ttau_s(T-s)}^{s,x}-g(t+\ttau_s(T-t)) \tilde F_{t+\ttau_s(T-t)}^{t,x}\right) \ind_{\lbrace\ttau_s<1\rbrace}\right.\\
			&  \quad\quad\left.\quad\quad\quad +\left(\max(G,F_{s+\ttau_s(T-s)}^{s,x})-\max(G,F_{t+\ttau_s(T-t)}^{t,x})\right)\ind_{\lbrace\ttau_s=1\rbrace}\right]\\
			&  \quad \leq \e\left[\left|F_{t+\ttau_s(T-t)}^{t,x}-F_{s+\ttau_s(T-s)}^{s,x}\right|\right]\\
			&  \quad \leq\e\left[xe^{\sigma\sqrt{T}M_1+|r-\sigma^2/2|T+\int_0^T c(u)\diff u} \Big|(r-\sigma^2/2)\ttau_s (s-t) 
			+\int_s^{s+\ttau_s(T-s)}c(u)\diff u\right.\\
			&  \quad\quad\left.-\int_t^{\ttau(t+T-t)}c(u)\diff u+\sigma\sqrt{T-t}W_{\ttau_s}-\sigma\sqrt{T-s}W_{\ttau_s}\Big|\right]\text{ (by \eqref{eqLipsW} and \eqref{eqLipsEXP})}\\
			&  \quad \leq\e\left[xe^{\sigma\sqrt{T}M_1+|r-\sigma^2/2|T+\int_0^T c(u)\diff u}\Big\lbrace \left|r-\sigma^2/2\right|\ttau_s |t-s|\right.\\
			&  \quad\quad\quad\quad\left.+3|t-s|+\sigma |W_{\ttau_s}| |\sqrt{T-t}-\sqrt{T-s}|\Big\rbrace\right]\qquad\textrm{(by \eqref{eq:boundIntcu})}\\
			& \qquad \leq |t-s|  (|r-\sigma^2/2| + 3) xe^{|r-\sigma^2/2|T+\int_0^T c(u)\diff u}
			\e\left[xe^{\sigma\sqrt{T}M_1}\right] \\
			&\qquad\qquad + |\sqrt{t-s}| x \sigma e^{|r-\sigma^2/2|T+\int_0^T c(u)\diff u}  
			\e[M_1 e^{\sigma\sqrt{T}M_1}]\\
			&  \quad = \tilde C_1|t-s|+C_2|\sqrt{t-s}|,
		\end{align*}
		where 
		$\tilde C_1 = (|r-\sigma^2/2|+3) x e^{|r-\sigma^2/2|T+\int_0^T c(u)\diff u}
		\e\left[xe^{\sigma\sqrt{T}M_1}\right]$ 
		and
		$C_2 = x \sigma e^{|r-\sigma^2/2|T+\int_0^T c(u)\diff u}  
		\e[M_1 e^{\sigma\sqrt{T}M_1}]$.
		
		Case 2: $v(s,x) < v(t,x)$. Note that the surrender function $g$ is Lipschitz, since it is in $C^1$ and it has bounded first-order derivatives (by hypothesis). Hence, we have
		\begin{align*}
			\allowdisplaybreaks
			& |v(t,x)-v(s,\,x)|=v(t,x)-v(s,\,x) \\
			&\quad \leq \e\left[e^{-r\ttau_t(T-t)}\varphi\left(t+\ttau_t(T-t), \tilde F_{t+\ttau_t(T-t)}^{t,x}\right)\right.\\
			&\quad\quad\quad\quad \left.-e^{-r\ttau_t(T-s)}\varphi\left(\ttau_t(T-s)+s, \tilde F_{s+\ttau_t(T-s)}^{s,x}\right)\right]\\
			&\quad= \e\left[\left|e^{-r\ttau_t(T-t)}-e^{-r\ttau_t(T-s)}\right|\varphi\left(t+\ttau_t(T-t), \tilde F_{t+\ttau_t(T-t)}^{t,x}\right)\right]\\
			&\quad\quad\quad\quad+\e\left[e^{-r\ttau_t(T-s)}\left\{\varphi\left(t+\ttau_t(T-t), \tilde F_{t+\ttau_t(T-t)}^{t,x}\right)-\varphi\left(\ttau_t(T-s)+s, \tilde F_{s+\ttau_t(T-s)}^{s,x}\right)\right\}\right]\\
			&\quad\leq \e\left[e^{rT}|r\ttau_t (T-t)-r\ttau_t(T-s)|\varphi\left(t+\ttau_t(T-t), \tilde F_{t+\ttau_t(T-t)}^{t,x}\right)\right]\quad\textrm{ (by \eqref{eqLipsEXP})}\\
			&\quad\quad\quad\quad +\e\left[e^{-r\ttau_t(T-s)}\left(g(t+\ttau_t(T-t)) \tilde F_{t+\ttau_t(T-t)}^{t,x}-g(s+\ttau_t(T-s)) \tilde F_{s+\ttau_t(T-s)}^{s,x}\right)\ind_{\lbrace\ttau_t<1\rbrace}\right.\\
			&\quad \left.\quad\quad\quad\quad +e^{-r\ttau_t(T-s)}\left(\max(G,\tilde F_{t+\ttau_t(T-t)}^{t,x})-\max(G,\tilde F_{s+\ttau_t(T-s)}^{s,x})\right)\ind_{\lbrace\ttau_t=1\rbrace}\right]\\
			&\quad\leq re^{2rT} \left(G + x\right) |t-s| \quad\textrm{ (by Lemma \ref{lemmaVprop1})}\\
			& \quad\quad\quad\quad +\e\Bigg[e^{-r\ttau_t(T-s)}\Big(\left(g(t+\ttau_t(T-t))-g(s+\ttau_t(T-s))\right) \tilde F_{s+\ttau_t(T-s)}^{s,x}\\
			& \quad\quad\quad\quad\quad +g(t+\ttau_t(T-t))\left[ \tilde F_{t+\ttau_t(T-t)}^{t,x}- \tilde F_{s+\ttau_t(T-s)}^{s,x}\right]\Big)\ind_{\lbrace\ttau_t<1\rbrace}\\
			& \quad\quad\quad\quad\quad +e^{-r\ttau_t(T-s)}\left|\tilde F_{t+\ttau_t(T-t)}^{t,x} - \tilde F_{s+\ttau_t(T-s)}^{s,x} \right|\ind_{\lbrace\ttau_t=1\rbrace}\Bigg]\\
			&\quad\leq re^{2rT} \left(G + x\right) |t-s| 
			+ \e\left[e^{-r\ttau_t(T-s)}C|t-s| \tilde F_{s+\ttau_t(T-s)}^{s,x}\right]\\
			&\quad\quad\quad\quad + \e\left[\left|\tilde F_{t+\ttau_t(T-t)}^{t,x}- \tilde F_{s+\ttau_t(T-s)}^{s,x}\right|\right]\\
			&\quad \leq (re^{2rT}(G+x) + Cx + \tilde C_1) |t-s| + C_2 \left|\sqrt{T-s} - \sqrt{T-t}\right| \\
			&\quad \leq C_1 |t-s| + C_2 \left|\sqrt{t-s}\right|
		\end{align*}
		where $C$ is the Lipschitz constant of $g$, $\tilde C_1$ and $C_2$ are the constants defined in Case 1 and $C_1 = re^{2rT}(G+x) + Cx + \tilde C_1$.
	\end{enumerate}
	
\end{proof}

\bibliographystyle{spbasic}      

\bibliography{myBib}

\end{document}